\documentclass[12pt, letterpaper]{article}

\pdfoutput=1
\usepackage{graphicx}
\usepackage{amsmath}
\usepackage[margin=1in]{geometry}
\usepackage[stretch=10]{microtype}
\usepackage{amssymb}
\usepackage{booktabs}
\usepackage{array}
\usepackage[utf8]{inputenc}
\usepackage{color}
\usepackage{bm}
\usepackage{multirow}

\usepackage{longtable}
\usepackage{wrapfig}
\usepackage{float}
\usepackage{colortbl}
\usepackage{pdflscape}
\usepackage{tabu}
\usepackage{threeparttable}
\usepackage{threeparttablex}
\usepackage[normalem]{ulem}
\usepackage{makecell}
\usepackage{xcolor}
\usepackage{blkarray}
\usepackage{bigstrut}

\usepackage{algorithm}
\usepackage{algpseudocode}

\usepackage[affil-it]{authblk}

\usepackage{nomencl}
\makenomenclature

\usepackage{amsthm}
\newtheorem{theorem}{Theorem}[section]
\newtheorem{lemma}[theorem]{Lemma}

\algnewcommand{\Initialize}[1]{%
  \State \textbf{Initialize:}
  \Statex \hspace*{\algorithmicindent}\parbox[t]{.8\linewidth}{\raggedright #1}
}

\usepackage{lmodern}
  \usepackage[T1]{fontenc}
  \usepackage{textcomp} 
  \usepackage{iftex}

\usepackage{amssymb}
  \providecommand{\keywords}[1]
{
  \small	
  \textbf{\textit{Keywords---}} #1
}

\title{Bayesian Event Categorization Matrix Approach for Explosion Monitoring}
\author[1]{Scott Koermer \thanks{Corresponding author, skoermer@lanl.gov}}
\author[2]{Joshua D. Carmichael}
\author[1]{Brian J. Williams}

\affil[1]{Statistical Sciences Group, Los Alamos National Laboratory, Los Alamos, NM, USA}
\affil[2]{National Security Earth Science Group, Los Alamos National Laboratory, Los Alamos, NM, USA}
\date{\today}
\usepackage[citestyle=authoryear,natbib=true,backend=bibtex]{biblatex}
\bibliography{bib.bib}
\usepackage{etoolbox}
\renewcommand\nomgroup[1]{%
  \item[\bfseries
  \ifstrequal{#1}{B}{Model parameters}{%
  \ifstrequal{#1}{P}{Prior hyperparameters}{%
  \ifstrequal{#1}{A}{Algorithms}{%
  \ifstrequal{#1}{T}{Decision theory}{%
  \ifstrequal{#1}{D}{Discriminants}{%
  \ifstrequal{#1}{M}{Mathematical concepts}{%
  \ifstrequal{#1}{O}{Other}{%
  \ifstrequal{#1}{N}{Distributions}{%
  \ifstrequal{#1}{I}{Dimension and Indexing}{}}}}}}}}}%
]}
\begin{document}
\maketitle

%\section*{Executive Summary}

\begin{abstract}
Current efforts to correctly categorize natural events from suspected explosion sources with data that is collected by ground- or space-based sensors presents historical challenges that remain unaddressed by the Event Categorization Matrix (ECM) model. Smaller historical events (lower yield explosions) often include only sparse observations among few modalities and can therefore lack a complete set of discriminants. The covariance structures can also vary significantly between such observations of event (source-type) categories. Both obstacles are problematic for the ``classic'' Event Categorization Matrix model. Our work addresses this gap and presents a Bayesian update to the previous Event Categorization Matrix model, termed the Bayesian Event Categorization Matrix model, which can be trained on partial observations and does not rely on a pooled covariance structure. We further augment the Event Categorization Matrix model with Bayesian Decision Theory so that false negative or false positive rates of an event categorization can be reduced in an intuitive manner. To demonstrate improved categorization rates for the Bayesian Event Categorization Matrix model, we compare an array of Bayesian and classic models with multiple performance metrics using Monte Carlo experiments. We use both synthetic and real data. Our Bayesian models show consistent gains in overall accuracy and a lower false negative rates relative to the classic Event Categorization Matrix model. We propose future avenues to improve Bayesian Event Categorization Matrix models for further improving decision-making and predictive capability.
\end{abstract}

\keywords{
    Bayesian inference,
    Monte Carlo methods,
    Probability distributions,
    Statistical methods
  }

\section{Introduction}

%\begin{quote}
%The central question in the discussion of an adequate monitoring of a [comprehensive test ban treaty] is that of %confidently identifying seismic events, i. e., to judge whether observed seismic signals are generated by an %underground nuclear explosion or by an earthquake. 
%
%\hfill  \citetitle{dahlman2016monitoring}\\
%\hspace*{\fill}\citet{dahlman2016monitoring}
%\end{quote}
Statistical methods have historically supported monitoring signatures of suspected conventional and nuclear explosions \citep{Bowers2009_1, Anderson2010_1}. Explosion monitoring researchers, in particular, leverage such methods to confirm or challenge the hypotheses that some geophysical events present evidence of a detonation, rather than natural processes \citep{National2012_1, Mcgrath2009_1}. Such methods have been crucial in seismic source identification, that is, statistical methods that screen explosion-sourced records of seismic activity from records expected from earthquakes or other processes. 

Early quantitative work \citep{Booker1964_1} that developed classification methods to separate explosions from earthquake populations (discrimination) in the 1960s justified later efforts to rigorously defend test-ban treaties \citep{Ericsson1970_1}. Some of these treaties were only aspirational at the time that the geophysical work was achieved \citep{Myers1972_1, Elvers1974_1}. More modern efforts from researchers like Shumway \citep{Shumway1984_1, Shumway1996_1, Shumway2001_1}, Anderson \citep{Anderson2010_2, Anderson2010_1, Anderson2014_1, Anderson2009_1}, and their coworkers \citep{Jih1990_1} have further advanced these statistical methods beyond discrimination. Such newer methods often ingest multiple discriminants \citep{Fah2002_1, anderson2007mathematical} or other modalities \citep{Redman2019_1} in statistical tests that screen all-source explosions from earthquakes. Some recent research has continued the trend to use data integration or fusion methods to more confidently detect \citep{Scoles2020_1, Carmichael2020_1, Carmichael2016_1}, identify \citep{Taylor2010_1, Arrowsmith2013_1}, and characterize \citep{Ford2021_1, Ford2014_1, Green2013_1, Williams2021_1} populations of explosions from other events, and thereby reduce false positive rates. This effort continues to focus on smaller, evasively conducted explosions \citep{koper2020importance, Rodd2023_1}.

One such multi-discriminant, statistical method that supports explosion monitoring is called the Event Categorization Matrix \citep[ECM,][]{anderson2007mathematical}. This method has historically been used to test if multiple observations associated with a single event support the hypothesis that their source was a conventional or nuclear detonation \citep{maceira2017trends}. The ECM method currently consumes seismic discriminants like event depth, ratios of body-wave and surface-wave magnitudes, and ground motion polarity. Some variants of ECM also leverage more novel discriminants, like infrasound phase arrivals and teleseismic waveform complexity factors (CFs) \citep{Anderson2008_1}. Such populations of discriminants that are sourced by detonations largely segregate from populations that are sourced by nuisance events that form other categories, like shallow earthquakes and deep earthquakes. 

The ECM method assumes that a vector of discriminant observations can be modeled as a random variable from a multivariate normal distribution, with a mean and covariance matrix specific to a single event category. The ECM method uses previous observations with known event categories (ground truth data) to estimate mean and covariance parameters for each event category distribution, and applies regularized discriminant analysis \citep[RDA,][]{friedman1989regularized} to estimate covariance matrices for small data sets. The ECM model then categorizes a new observation with a series of hypothesis tests that are based on typicality indices \citep{mclachlan2005discriminant}, and quantifies the likely set membership of the new, uncategorized, observation to each candidate event group. If a new observation is atypical of all categories, with the exception of the detonation category, its source is then categorized as an explosion.

Technologies to monitor for nuclear explosions have historically leveraged multiple sensor networks like the VELA satellites \citep{Wright2017_1} and seismic arrays \citep{Ringdal1982_1}. However, research to routinely fuse such multi-modal, simultaneous observations and improve event categorization accuracy remains on-going \citep{herzog2017nuclear, kalinowski2023innovation}. Both physical and mathematical issues each challenge implementing ECM with multiple modalities \citet{anderson2007mathematical}. 

Firstly, it is operationally difficult to associate multi-modal signals to the same causative source for all but the largest events. Therefore, analyses of such data can produce limited sets of discriminants. This means that the ECM model, which must use the same set of discriminants that it is calibrated against to categorize a new event, cannot ingest data from a new event that contains only partial observations, relative to that calibration data. 

Secondly, ECM uses a covariance estimator which can lead to model mis-specification, because there is the potential for events that produce multi-modal signals to have a drastically different covariance structures. For example, high-yield, aboveground nuclear detonations will produce optical signals with high irradiance that covaries with large amplitude seismic waveforms \citep{ford2021joint}. This does not imply that a nuisance event that produces a high irradiance signal will also produce seismic waveforms with significant amplitude. Therefore, while a conventional or nuclear detonation may produce covarying discriminants, nuisance events may not. 

Lastly, ECM may categorize an event for non-intuitive reasons. This is particularly true when the number of event categories considered grows. When ECM then fails to reject a new observation from multiple event categories, the model declares the event as indeterminate, requiring human intervention for categorization \citep{anderson2007mathematical}. A monitoring agent must then explore strategies to empirically reduce the observed number of false negatives, while not significantly increasing the number of indeterminate categorizations or the overall categorization accuracy.  Some of the gaps in \citeauthor{anderson2007mathematical}'s \parencite*{anderson2007mathematical} approach are a result of hypothesis testing and parameter estimation procedures associated with classical statistical methods.  Hence, subsequent references to the ECM model provided in \citet{anderson2007mathematical} are termed classical ECM (C-ECM). 

\begin{figure}
\centering
\includegraphics[width=0.95\textwidth]
{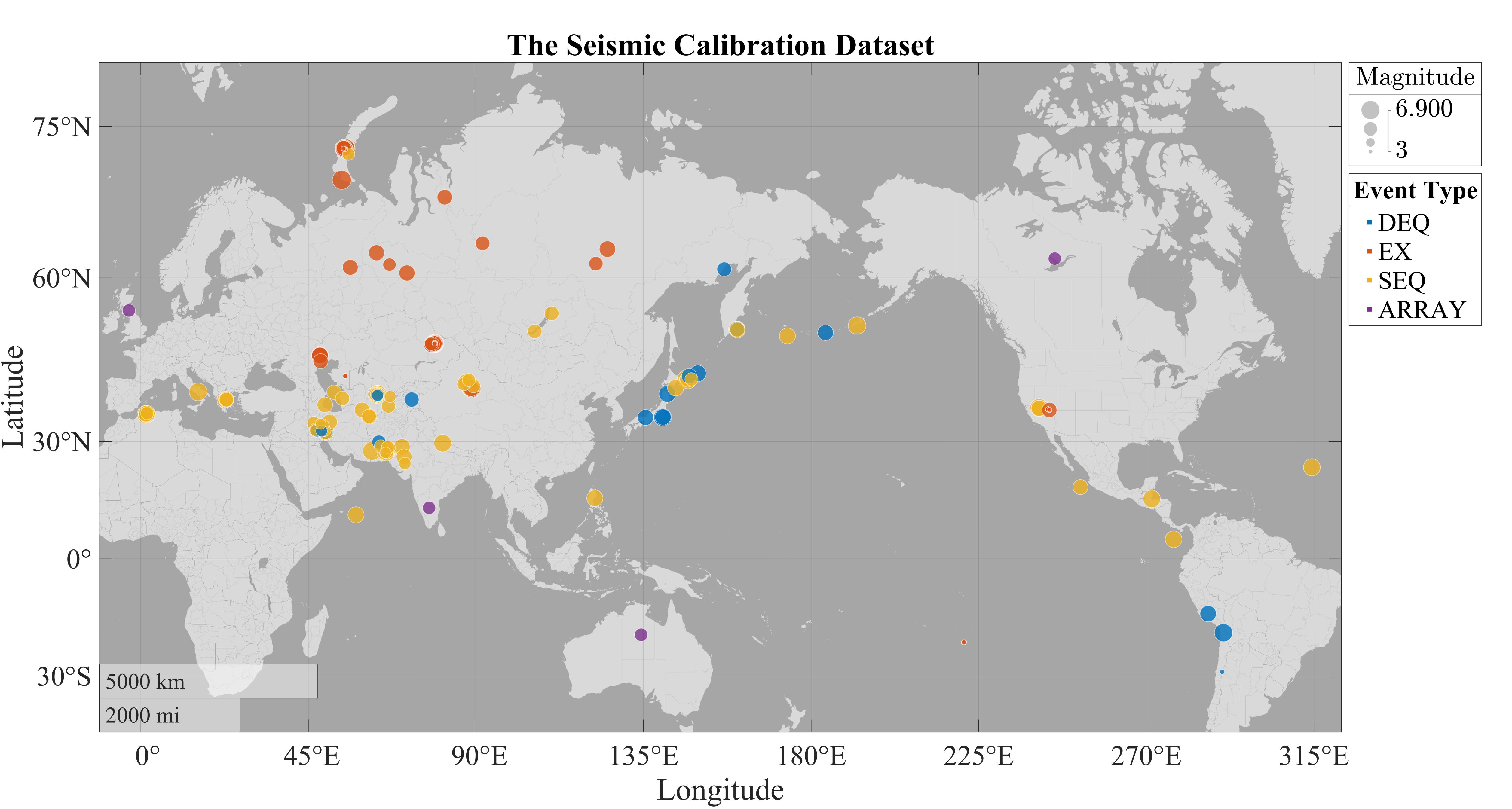}
\caption{The source locations for Deep Earthquake (DEQ), Explosion (EX), and Shallow Earthquake (SEQ) data used to calibrate C-ECM and B-ECM. Many sources were recorded from historical nuclear tests and their size is marked by seismic magnitude. Explosions with entirely missing magnitude data are marked with a ``magnitude 3'' size marker; this provides a crude estimate for global magnitude of completeness in most continental regions. The violet circles mark locations of the the four seismic arrays (labeled ``ARRAY'') that provided the bulk of the discriminant data.}
\label{fig:ecm-map}
\end{figure}

To address some of the problems faced with fusing discriminants from multiple modalities, we develop an ECM model which uses Bayesian methods (B-ECM) for parameter estimation and decision criteria. This novel work uses Bayesian decision theory and treats missing data with a matrix $t$-distribution within a Bayesian Normal mixture model \citep{stephens1997bayesian}. We evaluate integrals analytically when possible to avoid the computational expense of multi-dimensional numerical integration. These collective advances then select an event category for a new observation within practical time constraints and without supervision. We utilize a Bayesian typicality index to detect if a new observation is inconsistent with the event categories used for training, under the uncertainty imparted by using training data with missing elements. Lastly, we demonstrated the accuracy of the B-ECM methodology against a carefully prepared, curated dataset with missing entries (Fig. \ref{fig:ecm-map}), which C-ECM has previously been demonstrated against \citep{Anderson2008_1}. This establishes a baseline and measures gains in accuracy with the same curated dataset.

We organize this report as follows: Section \ref{sec:stat-basis} summarizes the statistical methods assimilated to create the B-ECM methodology. Section \ref{sec:implementation} overviews how these methods are implemented in codes and algorithms. Section \ref{sec:exp} details results from a series of Monte Carlo (MC) experiments that compare the performance of the B-ECM models versus the C-ECM model. Finally, section \ref{sec:discussion} provides a discussion of the results, their implications, and potential avenues for further improvements to event categorization models.

\section{Statistical Basis for the Decision Framework} \label{sec:stat-basis}
B-ECM is an end-to-end formalized framework that exploits geophysical discriminants to categorize an unknown event. We use Bayesian inference (Section \ref{sec:bayes}) throughout this work to ensure consistency between the computational and data processing stages of B-ECM. Our statistical model (Section \ref{sec:bayescat}) leverages Bayesian inference to use training data with missing discriminants. Bayesian decision theory (Section \ref{sec:decision-thy}) enables consistent decision making, taking into account category probabilities, the utility of a correct categorization, and the loss of an incorrect categorization.   Notation is detailed in Nomenclature Appendix \ref{sec:nomencl}.

\subsection{Bayesian Inference}\label{sec:bayes}
Bayesian statistics relies on Bayes' Theorem (Equation \ref{eq:bayesrule}) to infer model parameters. Bayes' rule reads as:  the posterior distribution of \(\phi\) given data \(y\); \(p(\phi|y)\), is equal to the likelihood of the data \(p(y|\phi)\) times the prior distribution on the model parameter \(p(\phi)\), divided by the marginal likelihood of the data \(p(y)\):
\begin{equation} \label{eq:bayesrule}
p(\phi|y) = \frac{p(y|\phi)p(\phi)}{p(y)}
\end{equation}
Bayesian inference treats probability as model parameter uncertainty, resulting in a probability distribution on \(\phi\) instead of a point estimate. Computational methods, like the Gibbs sampler, often use Markov chain Monte Carlo (MCMC) to infer model parameters by sampling model parameters from the posterior distribution $p\left( \phi \vert y\right)$ \citep{gelfand1990illustration, gelfand1990sampling, casella1992explaining, geman1984stochastic}. An observer can consider the probabilities for multiple models, given the data, to make data predictions under the Bayesian framework.  We use these traits of Bayesian inference to derive B-ECM. \citet{hoff2009first} provdies more resources on Bayesian inference  and \citet{robert1999monte} gives more details about MCMC and statistical simulation.

\subsection{Bayesian Categorization with Missing Training Data}\label{sec:bayescat}

We choose to categorize data with a methodology that is similar to that in \citet{stephens1997bayesian}. This method assumes that the event category for each training data event known, but the group for which a new observation belongs to is unknown. We let \(\bm{Y}_{N \times p}\) be a matrix that contains the entirety of the training data, with \(N\) event observations each with \(p\) observed discriminants.  Each row of \(\bm{Y}_{N\times p}\) is an observation of the \(k^{\mathrm{th}}\) event category, where \(k \in \{ 1, \dots, K\}\) has no uncertainty. We split the training data into the event category-specific matrices \(\bm{Y}_{N_1 \times p}, \dots , \bm{Y}_{N_k \times p}, \dots , \bm{Y}_{N_K \times p}\). Here, \(N_k\) is the number of training observations in the \(k^{\mathrm{th}}\) event category and \(N_1 + \dots + N_k + \dots + N_K = N\).  The \(p\) discriminants in each \(\bm{Y}_{N_k \times p}\) are the same and are arranged in the same column order.

\begin{figure}\label{fig:matrix-format}
\centering
\includegraphics[width =\columnwidth]{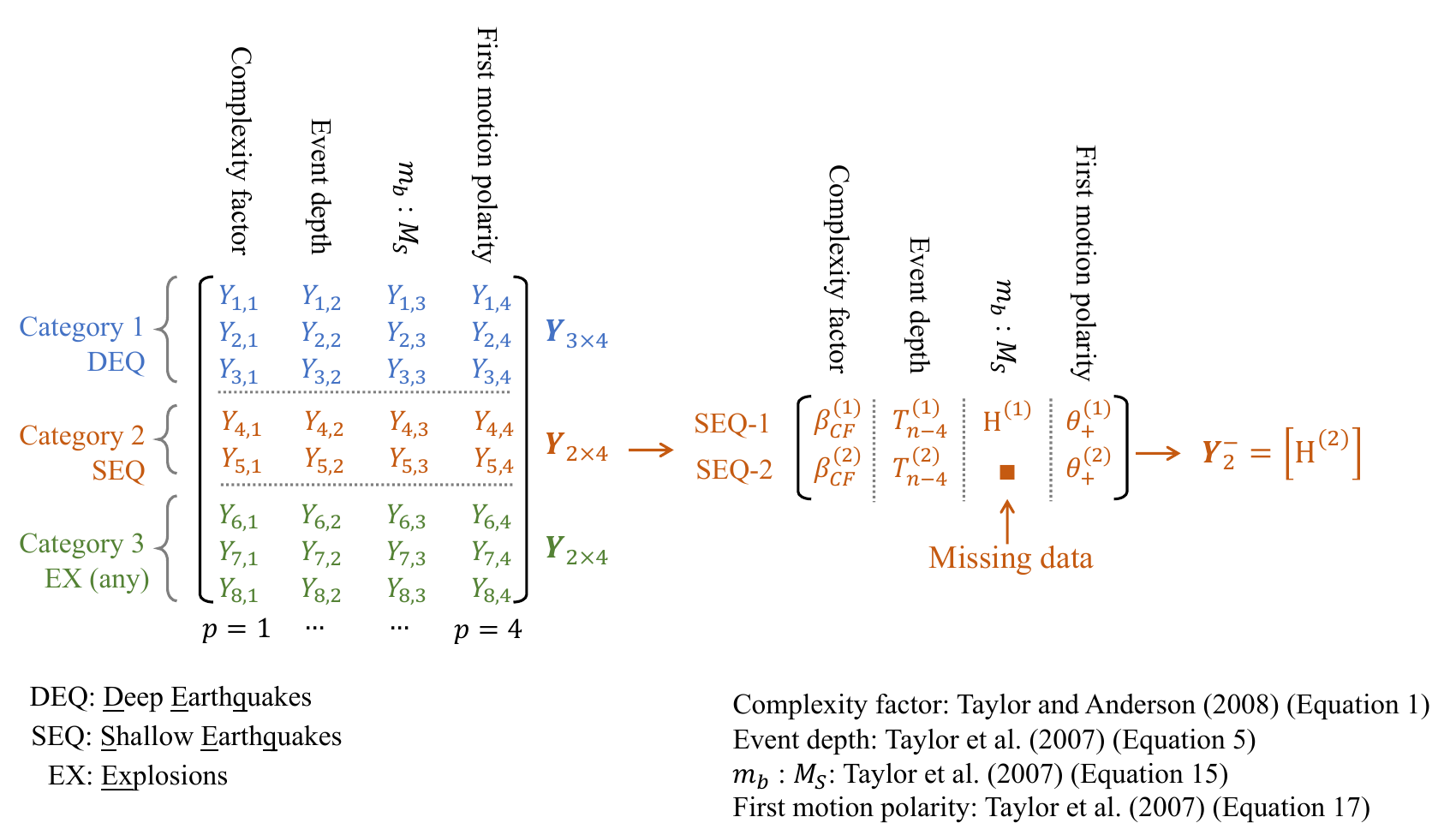}
\caption{A visual description of the training data \(\bm{Y}_{N \times p}\) that populates the B-ECM module. Dimensions \(N=8\) and \(p=4\) in this example. The top, leftmost column matrix organizes data within three categories (deep earthquakes, shallow earthquakes, and explosions) into distinct groups of rows. The four columns enumerate distinct seismic discriminant types (complexity factor, event depth, body wave to surface wave magnitudes, and polarity of first motion). The row matrices \(\bm{Y}_{N_k \times p}\) compartmentalize this training data by category, which we color separately. The sub-matrix \(\bm{Y}_{2 \times 4}\) \((N_2 = 2, p = 4)\) then stores data observed from shallow earthquakes (SEQ), according to the notation cited at the bottom right. The second shallow earthquake observation lacks a body wave to surface wave magnitude measurement. This missing data then populates the data set \(\bm{Y}_{2}^{-}\).}
\end{figure}

We assume that each \(\bm{Y}_{N_k \times p}\) is a realization from a Matrix Normal distribution \citep{gupta2018matrix}, with an unknown mean \(\bm{1}_{N_k} \bm{\mu}_k^{\top}\), independent rows, and column covariance \(\bm{\Sigma}_k\). Appendix \ref{sec:normal-likelihood} provides details. The elements of a Matrix Normal distributed random variable have infinite support (that is, the interval where the density is not identically zero) on the real number line. Previously, p-values on \((0,1]\) have been used with ECM \citep{anderson2007mathematical}.  We choose to transform these values with a logit transformation \(\mathrm{logit}(x) = \mathrm{ln}(x) - \mathrm{ln}(1-x)\) so that the transformed values are (instead) unbounded. The arcsine transform used in \citet{anderson2007mathematical} returns values on \([0,1]\) instead of on the entire real number line. This property is a mis-match for the normal distribution, which has support for all real numbers. However, both transforms are available for use in our code.

We integrate the Matrix Normal density of each \(\bm{Y}_{N_k \times p}\) over the prior densities of \(\bm{\mu}_k\) and \(\bm{\Sigma}_k\) for each \(k \) analytically. This reduces the computational burden of inference, as detailed in Appendix \ref{sec:data-likelihood}. Integration results in the marginal likelihood \(p(\bm{Y}_{N_k \times p}|\bm{\theta}_k)\), which is a matrix t-distribution \citep{gupta2018matrix}, conditioned on a set of prior hyperparameters, abbreviated as \(\bm{\theta}_k\) for the hyperparameters specific to group \(k\).  We refer to the set of prior hyperparameters, totaled over all \(K\) event categories, as \(\bm{\theta}\).  Appendix \ref{sec:priors} details the prior distributions and hyperparameters in \(\bm{\theta}\). Appendix \ref{sec:matrix-t} provides relevant properties of the matrix t-distribution.

Sometimes an event contained in \(\bm{Y}_{N_k \times p}\) will have less than \(p\) recorded discriminants, which is therefore unavailable. The unavailable discriminants for such a case are considered \emph{missing data} \citep{little2019statistical, gelman1995bayesian}. We let \(\bm{Y}_k^{-}\) represent the missing data elements from all \(N_k\) events in \(\bm{Y}_{N_k \times p}\) and \(\bm{Y}_k^{+}\) be the recorded elements, respectively. We leave \(\bm{Y}_k^{+}\) and \(\bm{Y}_k^{-}\) dimensionless to represent the observed and missing data with generality. We use the chain rule of conditional probability to specify 
\begin{equation}
p(\bm{Y}_{N_k \times p}|\bm{\theta}_k) = p(\bm{Y}_k^{-}, \bm{Y}_k^{+}|\bm{\theta}_k) = p(\bm{Y}_k^{-}| \bm{Y}_k^{+}, \bm{\theta}_k)p(\bm{Y}_k^{+}| \bm{\theta}_k).
\end{equation}
Because of the properties of conditional probability and \(p(\bm{Y}_{N_k \times p}|\bm{\theta}_k)\), we can obtain a large number of random draws, or samples, of \(\bm{Y}_k^{-}\) given \(\bm{Y}^{+}_k\) and \(\bm{\theta}_k\) from the probability density function \(p(\bm{Y}_k^{-}| \bm{Y}_k^{+}, \bm{\theta}_k)\) using a combination of conditional formulas (Appendices \ref{sec:matrix-t} and \ref{sec:missing-data-conditionals}) for the matrix t-distribution and numerical methods.  We use these probabilistic imputations of \(\bm{Y}_k^{-}\) to exploit partial observations in the training data set, and thereby compute more accurate decisions by utilizing instead of discarding partial observations.

We now use the imputed \(\bm{Y}_k^{-}\) and return to a \emph{full data set} \(\bm{Y}_{N_k \times p}\) and \(p(\bm{Y}_{N_k \times p}|\bm{\theta}_k)\) to next quantify the uncertainty that a new event belongs to each training category.  We then let \(\tilde{\bm{y}}_{\tilde{p}}\) be such a new event recorded as a vector of discriminants of length \(\tilde{p} \leq p\).  Vector \(\tilde{\bm{y}}_{\tilde{p}}\) is associated with the random vector of length \(K\) that we call \(\tilde{\bm{z}}_K^{\top} = [\tilde{z}_1, \dots, \tilde{z}_k, \dots , \tilde{z}_K]\).  A random realization of \(\tilde{\bm{z}}_K\) is equivalent to a draw from a multinomal distribution for \(n = 1\), where a single element is equal to 1 and the remaining elements are zero. The index, \(k \in \{1, \dots , K\}\), of the non-zero value of \(\tilde{\bm{z}}_K\), corresponds to \(\tilde{\bm{y}}_{\tilde{p}}\) belonging to the \(k^{\mathrm{th}}\) training event category. The conditional expected value of \(\tilde{\bm{z}}_K\), given all the data and prior parameter specifications, is then equivalent to: 
\begin{equation} \label{eq:expectz}
\mathbb{E}[\tilde{\bm{z}}_K|\tilde{\bm{y}}_{\tilde{p}}, \bm{Y}_{N \times p}, \bm{\theta}] = [p(\tilde{z}_1 = 1|\tilde{\bm{y}}_{\tilde{p}}, \bm{Y}_{N \times p}, \bm{\theta}), \dots , p(\tilde{z}_K = 1|\tilde{\bm{y}}_{\tilde{p}}, \bm{Y}_{N \times p}, \bm{\theta})].
\end{equation}
Equation \ref{eq:expectz} is equivalent to a vector that specifies the probability of \(\tilde{\bm{y}}_{\tilde{p}}\) belonging to each of the \(K\) event categories. For the \(k^{\mathrm{th}}\) event category, \(p(\tilde{z}_k = 1|\tilde{\bm{y}}_{\tilde{p}}, \bm{Y}_{N \times p}, \bm{\theta})\) can be evaluated using Bayes' rule.
\begin{equation}\label{eq:zpred}
p(\tilde{z}_k = 1|\tilde{\bm{y}}_{\tilde{p}}, \bm{Y}_{N \times p}, \bm{\theta}) = \frac{p(\tilde{\bm{y}}_{\tilde{p}}|\tilde{z}_k = 1, \bm{Y}_{N_k\times p}, \bm{\theta}_k)p(\tilde{z}_k = 1|\bm{Y}_{N \times p}, \bm{\theta})}{p(\tilde{\bm{y}}_{\tilde{p}}|\bm{Y}_{N\times p}, \bm{\theta})}
\end{equation}
All densities on the right hand side are available in closed form.  Appendix \ref{sec:predictive-y} details \(p(\tilde{\bm{y}}_{\tilde{p}}|\tilde{z}_k = 1, \bm{Y}_{N_k\times p}, \bm{\theta}_k)\), the predictive density for category \(k\). Appendix \ref{sec:predictive-a-priori-y} details \(p(\tilde{z}_k = 1|\bm{Y}_{N \times p}, \bm{\theta})\), the probability that \(\tilde{z}_k = 1\) prior to observing \(\tilde{\bm{y}}_{\tilde{p}}\). Appendix \ref{sec:marginal-predictive} details \(p(\tilde{\bm{y}}_{\tilde{p}}| \bm{Y}_{N \times p}, \bm{\theta})\), the predictive density for $\tilde{\bm{y}}_{\tilde{p}}$ marginalized over all event categories.

\subsection{Bayesian Decision Theory}\label{sec:decision-thy}

We now focus on event categorization, that is, placing \(\tilde{\bm{y}}_{\tilde{p}}\) into one of the \(K\) training event categories, using Bayesian Decision Theory\citep{robert2007bayesian, berger2013statistical}.

We call \(a_k\,,\, k \in \{1, \dots , K\}\) the action of placing \(\tilde{\bm{y}}_{\tilde{p}}\) into the \(k^{\mathrm{th}}\) event category .  Given the data and prior specifications, the choice of each action \(a\) has an associated loss. Loss is unavoidably subjective, and specified by a loss function that can be evaluated for each action. The action with the minimum expected loss is considered the lowest risk in this Bayesian setting.  We choose a simple loss function that specifies a loss matrix of constants \(\bm{C}\)
\begin{equation}
\bm{C}_{K \times K} = 
\begin{blockarray}{ccccl}
 & a_1 & \dots & a_K \\
\begin{block}{c[ccc]l}
\tilde{z}_1 = 1 & C_{1,1} & \dots & C_{1,K}\bigstrut[t] &  \\
\vdots & \vdots & \ddots & \vdots &  \\
\tilde{z}_K = 1 & C_{K,1} & \dots & C_{K,K}\bigstrut[b] &, \\
\end{block}
\end{blockarray}
\end{equation}
where each element of \(\bm{C}\) corresponds to the loss for each action that is indexed by the columns, for a value of the random variable \(\tilde{\bm{z}}_K\) indexing the rows. The elements of \(\bm{C}\) are ideally chosen with some thoughtful use of utility theory \citep{robert2007bayesian, berger2013statistical}. For a draw of \(\bm{\tilde{z}}_K\), the loss function is evaluated as \(\bm{L}(\tilde{\bm{z}}_K, \bm{a}) = \tilde{\bm{z}}_K^{\top} \bm{C}\), producing a vector of the same length as the number of actions. A simple loss function uses \(\bm{C}_{K \times K} = \bm{J}_K - \bm{I}_K\), a matrix of ones minus the identity matrix such that the diagonal of \(\bm{C}_{K \times K}\) is equal to zero.  Under this loss function, and using a draw of \(\tilde{\bm{z}}_K\) with \(\tilde{z}_2 = 1\), the loss of taking action \(a_2\) and placing \(\tilde{\bm{y}}_{\tilde{p}}\) in category two is equal to zero, while the loss of each of the remaining actions is equal to one. Taking the expectation of the loss function, with respect to the posterior distribution of random vector \(\tilde{\bm{z}}_{K}\), is therefore \(\mathbb{E}\left[\bm{L}(\tilde{\bm{z}}_K, \bm{a})\right] = \mathbb{E}\left[\tilde{\bm{z}}_K^{\top}|\tilde{\bm{y}}_{\tilde{p}}, \bm{Y}_{N\times p}, \bm{\theta}\right]\bm{C}\).

This Bayesian decision criterion is adaptable to an array of detonation detection scenarios. For binary decisions, an observer decides only whether \(\tilde{\bm{y}}_{\tilde{p}}\) is a detonation or not. Such cases are common in explosion monitoring. Then \(\bm{C}\) is a \(2 \times 2\) matrix,
\begin{equation}
\bm{C}_{2 \times 2} = 
\begin{blockarray}{cccl}
 & a_1  & a_2 & \\
\begin{block}{c[cc]l}
\tilde{z}_1 = 1 & C_{1,1}  & C_{1,2}\bigstrut[t] & \\
\tilde{z}_1 \neq 1 & C_{2,1} & C_{2,2}\bigstrut[b] &, \\
\end{block}
\end{blockarray}
\end{equation}
where \(\tilde{z}_1\) corresponds to the event of interest.  When \(C_{1,1} = C_{2,2} = 0\) and \(C_{1,2} = C_{2,1} = 1\) the matrix represents the 0-1 loss function in classical hypothesis testing \citep{robert2007bayesian}.  In Section \ref{sec:exp}, we primarily investigate binary categorization to showcase how changing an element of \(\bm{C}_{2\times 2}\) allows one to intuitively target false negatives or false positives. Appendix \ref{sec:binary-decision} details how reducing the action space to binary categorization allows several simplifications.

When training data has missing entries, the posterior expected loss is evaluated as \(\mathbb{E}\left[\tilde{\bm{z}}_K^{\top}|\tilde{\bm{y}}_{\tilde{p}}, \bm{Y}^{+}, \bm{\theta}\right]\bm{C} \). Appendix \ref{sec:expected-loss-appendix} details the practical use of Monte Carlo integration to approximate this marginal expectation. 

A critique of this categorization method is that if \(\tilde{\bm{y}}_{\tilde{p}}\) is not from one of the \(K\) training event categories, and is a true outlier, there will still be an action with the lowest expected loss, and \(\tilde{\bm{y}}_{\tilde{p}}\) will be placed into the wrong category.  We address this issue by utilizing a Bayesian typicality index in conjunction with Bayesian Decision Theory.  Our Bayesian typicality index places a Bayesian twist on the typicality index \citep{mclachlan2005discriminant} used in C-ECM, by utilizing the multivariate t predictive distribution of \(\tilde{\bm{y}}_{\tilde{p}}\) and Bayesian decision theory to decide on the action of rejection in the event that elements of the training data are missing.  In the event that \(\tilde{\bm{y}}_{\tilde{p}}\) is rejected from the event category selected as the minimum expected loss action, via the typicality index, then \(\tilde{\bm{y}}_{\tilde{p}}\) is considered an outlier, possibly belonging to an event category not included in the \(K\) training event categories.  Appendix \ref{sec:bayes-typicality-index} gives more details use of the typicality index in a Bayesian setting.

\section{Implementation}\label{sec:implementation}

An \textsf{R} package titled \texttt{ezECM} implements the model; it is a ``living package'' under active improvement. The \texttt{ezECM} package provides functions for loading data, training a B-ECM model, saving and loading training results, predicting the category of a new observation, decision making using Bayesian decision theory, as well as summarizing and plotting results.  The \texttt{ezECM} package also includes an implementation of C-ECM, which provides a baseline for comparing empirical results between models.

When there are no missing training data, evaluation of each \(p(\tilde{z}_k = 1|\tilde{\bm{y}}_{\tilde{p}}, \bm{Y}_{N \times p}, \bm{\theta})\) is straightforward.  When there are missing data entries, we implement a Gibbs sampler \citep{casella1992explaining} in \texttt{ezECM} to generate samples from each \(p(\bm{Y}_k^{-}|\bm{Y}_k^{+}, \bm{\theta}_k)\).  The Gibbs sampler is the only computationally intense aspect of training our B-ECM model.  Pseudocode for this operation is provided as Algorithm \ref{alg:training}. A user of the algorithm either supplies prior parameters \(\bm{\theta}\), or uses the default values in the package function.   The user also provides the total number of samples \(T\) to take of \(\bm{Y}^{-}\), and the number of burn-in samples \(B\), which are discarded under the assumption that the Markov chain has not converged to the target distribution within the first \(B\) iterations.  At the end of the algorithm \(T - B\) total samples are obtained.   Initial values for the missing entries \((\bm{Y}^{-})^{(1)}\), need to be set at the start of the algorithm.  In \texttt{ezECM} the initialized missing elements \((\bm{Y}_k^{-})^{(1)}\) are taken to be the column mean of the observed elements.  

The missing entries from each column of each event category are drawn, conditional on the remaining entries.  For the \(\ell^{\mathrm{th}} \in \{1, \dots , p\}\) column of \(\bm{Y}_{N_k \times p}\) a column permutation matrix \(\bm{P}^{C}_{k, \ell}\) swaps the columns of \(\bm{Y}_{N_k \times p}\), and a row permutation matrix \(\bm{P}^{R}_{k, \ell}\) swaps the rows of  \(\bm{Y}_{N_k \times p}\) such that \(\underline{\bm{Y}}_{N_k \times p} =  \bm{P}^{R}_{k, \ell} \bm{Y}_{N_k \times p} \bm{P}^{C}_{k, \ell}\).  The missing data in the \(\ell^{\mathrm{th}}\) column of \(\bm{Y}_{N_k \times p}\) is found in the first elements of the first column of \(\underline{\bm{Y}}_{N_k \times p}\), so that the missing data in column \(\ell\) can be drawn from the conditional distributions found in Appendix \ref{sec:missing-data-conditionals}; Appendix \ref{sec:the-model} details the notation found in Algorithm \ref{alg:training}. Our algorithm saves the draw and updates the corresponding elements of \(\bm{Y}_{N_k \times p}\) with these values. We repeat this process for all columns of \(\bm{Y}_{N_k \times p}\) with missing values in the original training data set, and then repeated for \(T\) iterations. The process approximates draws from the joint distribution \(p(\bm{Y}_k^{-}|\bm{Y}^{+}_k, \bm{\theta}_k)\) \citep{robert1999monte} for each \(k\). 
%%% Algorithm %%%
% https://tex.stackexchange.com/questions/56871/how-to-format-for-loop-for-printing-a-pseudo-code-listing
%https://tex.stackexchange.com/questions/84722/slight-change-in-algorithm
\begin{algorithm}
\caption{Joint Monte Carlo samples of \(\bm{Y}^{-}_k ; \forall k \in \{1 , \dots , K\}\)}\label{alg:training}
\begin{algorithmic}
\Require \(\bm{\eta}, \bm{\Psi}, \bm{\nu}, \bm{Y}^{+}, \bm{P}^{R}_{k, \ell}, \bm{P}^{C}_{k, \ell}\)
\Initialize{\((\bm{Y}^{-})^{(1)}, (N_{k}^m)^{\ell} \times T \ \mathrm{Matrices \ to \ store \ samples} \ (\bm{y}^{-}_{k, \ell})^{(t)}; 
\forall \ k \in \{1, \dots , K\}, \ell \in \{1 , \dots , p\}, t \in \{1, \dots , T\} \)}
\For{\(t \in \{1,  \dots , T\}\)}
    \For{\(k \in \{1 , \dots , K\}\)}
        \For{\(\ell \in \{1  , \dots , p\}\)}
        \State \(\underline{\bm{Y}}_{N_k \times p} \gets \bm{P}^{R}_{k , \ell} \bm{Y}_{N_k \times p} \bm{P}^{C}_{k, \ell}\)
        \State \(\underline{\bm{\eta}}_k \gets (\bm{P}^{C}_{k, \ell})^{\top} \bm{\eta}_k\)
        \State \(\underline{\bm{\Psi}}_k \gets (\bm{P}^{C}_{k, \ell})^{\top} \bm{\Psi}_k \bm{P}^{C}_{k, \ell}\)
        \State \(\bm{y}^{-}_{k,\ell} \sim p(\underline{\bm{y}}^{-}_{N_k^m \times 1}|\underline{\bm{y}}^{+}_{N_k^o \times 1}, \underline{\bm{Y}}_{N_k^m \times (p-1)}, \underline{\bm{Y}}_{N_k^o \times (p-1)}, \bm{\eta}_k, \bm{\Psi}_k, \nu_k)\)
        \State \((\bm{y}_{k,\ell}^{-})^{(t)} \gets \bm{y}_{k,\ell}^{-}\)
        \State \( \underline{\bm{y}}^{-}_{N_k^m \times 1}  \gets \bm{y}_{k,\ell}^{-}\)
        \State \(\bm{Y}_{N_k \times p} \gets  (\bm{P}^{R}_{k , \ell})^{\top}  \underline{\bm{Y}}_{N_k \times p} (\bm{P}^{C}_{k, \ell})^{\top}\)
        \EndFor
    \EndFor
\EndFor
\end{algorithmic}
\end{algorithm}

Once we make draws of \(\bm{Y}^{-}\), we then approximate \(\mathbb{E}[\tilde{\bm{z}}_K|\tilde{\bm{y}}_{\tilde{p}}, \bm{Y}^{+}, \bm{\theta}]\) with a new observation \(\tilde{\bm{y}}_{\tilde{p}}\) with functions in \texttt{ezECM}, and input these decisions to the Bayesian decision theory framework (see Section \ref{sec:decision-thy}). Algorithm \ref{alg:prediction} documents pseudocode for this process. Using the \(t \in \{1 , \dots , T - B \}\) draws \((\bm{Y}^{-})^{(t)}\) Algorithm \ref{alg:training} outputs, we must evaluate the expected predictive category probability \eqref{eq:zpred} for each \(k \in \{1 , \dots , K\}\) by first joining \((\bm{Y}^{-})^{(t)}\) with observations \(\bm{Y}^{+}\) and evaluating the multivariate t-distribution density \(p(\tilde{\bm{y}}_{\tilde{p}}|\tilde{z}_k = 1, \bm{Y}^{+}, (\bm{Y}^{-})^{(t)}, \bm{\theta}_k)\) detailed in Appendix \ref{sec:predictive-y} for all \(k\).  Then the integral over \(\bm{Y}^-\) to find each density \(p(\tilde{\bm{y}}_{\tilde{p}}|\tilde{z}_k = 1, \bm{Y}_k^{+}, \bm{\theta}_k)\) is approximated \citep{ULAMS1949TMCM} as the mean over all \(t\), and used with the result in Appendix \ref{sec:expected-loss-appendix} to evaluate the expected category probabilities for a \(\tilde{\bm{y}}_{\tilde{p}}\).  In the case where \(\tilde{p} < p\) we use the properties of the marginal matrix t-distribution (Appendix \ref{sec:matrix-t}) to evaluate \(p(\tilde{\bm{y}}_{\tilde{p}}|\tilde{z}_k = 1, \bm{Y}^{+}_k, (\bm{Y}^{-}_{k})^{(t)}, \bm{\theta}_k)\).
\begin{algorithm}
\caption{\(\mathbb{E}[\tilde{\bm{z}}_K|\tilde{\bm{y}}_{\tilde{p}}, \bm{Y}^{+}, \bm{\theta}]\)}\label{alg:prediction}
\begin{algorithmic}
\Require \(\tilde{\bm{y}}_{\tilde{p}}, \bm{\theta} , \bm{Y}^{+}, (\bm{Y}^{-})^{(B, \dots , T)}\)
\Initialize{Data Matrix \(\bm{p}_{K \times (T-B)}\) to store samples \(p_k^{(t)} = p(\tilde{\bm{y}}_{\tilde{p}}|\tilde{z}_k = 1, (\bm{Y}_{N_k\times p})^{(t)}, \bm{\theta}_k)\)}
\For{\(k \in \{1 , \dots , K\}\)}
    \For{\(t \in \{B , \dots , T\}\)}
        \State \((\bm{Y}_{N_k \times p})^{(t)} \gets \bm{Y}_k^+ \cup (\bm{Y}_k^{-})^{(t)}\)
        \State \(p_k^{(t)} \gets p(\tilde{\bm{y}}_{\tilde{p}}|\tilde{z}_k = 1, (\bm{Y}_{N_k\times p})^{(t)}, \bm{\theta}_k)\)
    \EndFor
    \State \(p(\tilde{\bm{y}}_{\tilde{p}}|\tilde{z}_{k} = 1, \bm{Y}_k^{+}, \bm{\theta}_k) \gets \frac{1}{T-B} \sum_{t = B}^{T} p_k^{(t)}\)
\EndFor
\State \(p(\tilde{\bm{y}}_{\tilde{p}} | \bm{Y}^{+}, \bm{\theta}) \gets \sum_{k = 1}^{K} p(\tilde{\bm{y}}_{\tilde{p}}|\tilde{z}_{k} = 1, \bm{Y}_k^{+}, \bm{\theta}_k)p(\tilde{z}_k = 1|\bm{Y}_{N \times p},\bm{\theta})\)
\For{\(k \in \{1 , \dots , K\}\)}
\State \(\mathbb{E}[\tilde{z}_{k} = 1| \tilde{\bm{y}}_{\tilde{p}}, \bm{Y}^{+}, \bm{\theta}] = \frac{p(\tilde{\bm{y}}_{\tilde{p}}|\tilde{z}_{k} = 1, \bm{Y}_k^{+}, \bm{\theta}_k)p(\tilde{z}_k = 1|\bm{Y}_{N \times p},\bm{\theta})}{p(\tilde{\bm{y}}_{\tilde{p}} | \bm{Y}^{+}, \bm{\theta})}\)
\EndFor
\end{algorithmic}
\end{algorithm} 

The user can choose to trade a reduction in autocorrelation between the Monte Carlo samples of \(\bm{Y}^{-}\) for an increase in computation time using thinning in the predict function of \texttt{ezECM}. Thinning the samples by integer factor \(q\) then utilizes only every \(q^{\mathrm{th}}\) sample of \(\bm{Y}^{-}\) to execute Algorithm \ref{alg:prediction}.  The size of the integer set \(\{B , \dots , T\}\) used as values for the index \(t\) and divisor for computing the mean of \(p(\tilde{\bm{y}}_{\tilde{p}}|\tilde{z}_k = 1, \bm{Y}^{+}_k, (\bm{Y}^{-}_{k})^{(t)}, \bm{\theta}_k)\) in Algorithm \ref{alg:prediction} are adjusted accordingly.

Lastly, we provide a function in \texttt{ezECM} to evaluate the loss function and to find the minimum loss action, given the loss matrix \(\bm{C}_{2 \times 2}\) or \(\bm{C}_{K \times K}\). We also specify a category of importance. If the minimum loss action is to categorize \(\tilde{\bm{y}}_{\tilde{p}}\) as the specified category, then we calculate the typicality index of that category for a significance level \(\tilde{\alpha}\). If the algorithm deems \(\tilde{\bm{y}}_{\tilde{p}}\) as atypical of the category, then we consider \(\tilde{\bm{y}}_{\tilde{p}}\) to be an outlier, pending further analysis.

\section{Experiments}\label{sec:exp}
We now perform a series of two Monte Carlo experiments to quantify any advantages of B-ECM over C-ECM, the first using synthetic statistically generated data and the second using real ground based data. In each experiment we use a set of training data to fit the models, and use a set of testing data with a known truth to measure accuracy, false negatives, and false positives. We fit the five models that include 1) C-ECM, which can only utilize complete data records, 2) B-ECM with only complete data records, 3) B-ECM with all data records (M-B-ECM), 4) M-B-ECM with a loss function chosen to reduce false negatives (M-B-ECM \(C_{1,2} = 2\)), and 5) M-B-ECM for event categorization (M-B-ECM Cat).  Decision criteria in models 1, 2, 3, and 4 are set up for binary categorization. We used 0-1 loss for B-ECM variants, except model 4. In this case, we utilized the loss matrix that Equation \ref{eq:loss-false-pos} shows.
\begin{equation}\label{eq:loss-false-pos}
\bm{C}_{2 \times 2} = 
\begin{blockarray}{cccl}
 & a_1  & a_2 & \\
\begin{block}{c[cc]l}
\tilde{z}_1 = 1 & 0  & 2\bigstrut[t] & \\
\tilde{z}_1 \neq 1 & 1 & 0\bigstrut[b] & \\
\end{block}
\end{blockarray}
\end{equation}
We chose a level of significance for all typicality indices to be \(\tilde{\alpha} = 0.05\), for both C-ECM and B-ECM.  Decisions using any B-ECM model first require evaluating the posterior expected loss for each action considered.  For binary decisions, the minimum expected loss action is used to place \(\tilde{\bm{y}}_{\tilde{p}}\) in a presumptive category.  If the presumptive category is a detonation, the Bayesian typicality index (Appendix \ref{sec:bayes-typicality-index}) is used to check if \(\tilde{\bm{y}}_{\tilde{p}}\) is rejected as an extreme value for detonations.  Then, if \(\tilde{\bm{y}}_{\tilde{p}}\) is rejected from the detonation distribution, \(\tilde{\bm{y}}_{\tilde{p}}\) is categorized as not a detonation. Conversely, if there is failure to reject from the detonation distribution, \(\tilde{\bm{y}}_{\tilde{p}}\) is categorized as a detonation.  If the presumptive category is not a detonation in the first place, no typicality index is used and the presumptive category is the final categorization.

When using a B-ECM model for full \(K\) categorization, we calculate the typicality index for the minimum expected posterior loss action, both detonations and non-detonations.  In the case of rejection, \(\tilde{\bm{y}}_{\tilde{p}}\) is categorized as an outlier. No outlier categories were included in the data for our experiments. This means that categorization of \(\tilde{\bm{y}}_{\tilde{p}}\) as an outlier has a detrimental effect on accuracy, as well as false negatives when the true category for \(\tilde{\bm{y}}_{\tilde{p}}\) is detonations.

The calculation of accuracy takes on a different meaning between the binary categorization and the full \(k \in \{1, \dots , K\}\) categorization.  Failure to place \(\tilde{\bm{y}}_{\tilde{p}}\) in exactly the correct category is an inaccuracy for full \(K\) categorization.  In binary categorization, the event is either a detonation or ``something else''.  There is no penalization in accuracy for categorizing \(\tilde{\bm{y}}_{\tilde{p}}\) as a deep earthquake if in reality \(\tilde{\bm{y}}_{\tilde{p}}\) is a shallow earthquake.  Both are grouped together as a new ``something else'' category, and correctly guessing \(\tilde{\bm{y}}_{\tilde{p}}\) is something other than a detonation meets the requirements for a binary detection model.

We implemented B-ECM for both sets of experiments using the code provided in the \textsf{R} package \texttt{ezECM} with \(p(\tilde{\bm{z}}_K|\bm{Y}_{N \times p}, \bm{\theta})\) informed by the data, instead of being equally weighted over \(K\).  The logit function, \(\mathrm{logit}(p) = \ln(p) - \ln(1-p) \), was used to map p-values to \((-\infty, \infty)\).  The defalut prior parameters, selected to allow for wide ranging data observations, in the \texttt{ezECM} package were used.  Details on the prior parameters can be found in Appendix \ref{sec:priors}.

The B-ECM models that utilized data with missing entries, in which Monte-Carlo was required for inference, generated 50,500 draws of each \(\bm{Y}^{-}_k\) and discarded the first 500 draws as burn-in. These same implementations of B-ECM ``thinned'' the draws when we made predictions, and only utilized every \(5^{\mathrm{th}}\) draw, or 10,000 draws in total.

Our implementation of C-ECM in the experiments is identical to our implementation of C-ECM in the \texttt{ezECM} function. In particular, C-ECM fits the RDA model using the \texttt{klaR::rda()} function from the \texttt{klaR} package \citep{klaR2005} in \textsf{R}. We only used C-ECM to form binary decisions. We calculated typicality indices from the Mahalanobis distance by leveraging the methodology in \citet{anderson2007mathematical}. We structured the decision framework such that indeterminate and undefined categorizations could not occur, which the authors felt would unfairly detriment performance metrics for C-ECM if these were possible results from analysis. When there is a failure to reject \(\tilde{\bm{y}}_{\tilde{p}}\) from multiple event categories, the categorization is indeterminant.  If \(\tilde{\bm{y}}_{\tilde{p}}\) is rejected from all event categories, the categorization is undefined. First, the typicality indices were calculated for all event categories.  If \(\tilde{\bm{y}}_{\tilde{p}}\) was rejected from the detonation category, then \(\tilde{\bm{y}}_{\tilde{p}}\) was declared to be not a detonation. If \(\tilde{\bm{y}}_{\tilde{p}}\) was not rejected from the detonation category, the observation was rejected from all other categories to have been declared a detonation.  Alternatively, if \(\tilde{\bm{y}}_{\tilde{p}}\) was not rejected from the detonation category and one or more additional categories, the observation was declared to be not a detonation.

\subsection{Synthetic Data}\label{sec:synthetic-exp}

Our first experiment used synthetically generated data in 250, independent MC experiments for each of \(p \in \{4, 6, 8, 10\}\).  The data generating mechanism was designed to mimic what we would expect to see when fusing ground and discriminants recorded in space, such as event depth and the number of satellites reporting an event. Equation (\ref{eq:cov-question}) shows that we expect that for a \(\bm{y}_{p}\) from a given event category, that some discriminants will only have correlations within the space (S) and ground (G) modalities, while other event categories will have full correlation across discriminants.  Additionally, we expect the number of discriminants used in practice to be of moderate size, but not extremely large (fewer than 10). The values used for \(p\) reflect this choice
\begin{equation}\label{eq:cov-question}
\mathbb{C}\mathrm{ov}(\bm{y}_{p \times 1}) \overset{?}{=}  \left[ \begin{array}{cc}
\bm{\Sigma}_{\mathrm{S}} & \bm{0} \\
\bm{0} & \bm{\Sigma}_{\mathrm{G}}
\end{array} \right] \overset{?}{=}  \left[ \begin{array}{cc}
\bm{\Sigma}_{\mathrm{S}} & \bm{\Sigma}_{\mathrm{S, G}}\\
\bm{\Sigma}_{\mathrm{G, S}} & \bm{\Sigma}_{\mathrm{G}}
\end{array} \right].
\end{equation}
We used \(K = 3\) event categories to generate data. Each category had a unique randomly generated mean and covariance for each experiment. The number of observations in the sum of training and testing data was randomly generated from a multinomial distribution with \((N^{\mathrm{train}} + N^{\mathrm{train}^{-}} + N^{\mathrm{test}})\) trials and equal event probabilities (1/3 for each category). We randomly chose a subset of the event categories to have a block covariance structure with a random block size. We generated data from the draws of the multivariate normal distribution, using the set of event category specific randomly generated parameters, and then transformed onto \([0,1]^p\) using the logistic function so that the data mimics the p-values used in application.  The training data set contained full observations of 25 events. An additional 125 events, for which 50\% of the \(125 \times p\) elements were unobserved, were included in the training data set. Our testing data set included 100 events with 50\% of the \(p \times 100\) entries missing. Algorithm \ref{alg:syn-data} of Appendix \ref{sec:syn-data-appendix} details the data generating mechanism.

Fig. \ref{fig:syn-boxplot} shows the distribution of the observed accuracy for each of the 250 MC iterations  as a box-plot. Here, the binary categorization B-ECM models often performs better than C-ECM, even without taking advantage of partial observations for training. This difference is more pronounced as \(p\) increases. We hypothesize that C-ECM returned little change in median accuracy for increasing \(p\) because the decision criteria is relatively stringent, often leading to indeterminant or undefined categorizations.  A C-ECM model applied on a similar data set would require much additional human intervention.  Additionally, as \(p\) increases there is a growing gap in the variability of the observed accuracy between all B-ECM models and C-ECM.  From this observation, we infer that B-ECM can better leverage larger numbers of discriminants to deliver more consistent results than C-ECM.  

M-B-ECM Cat tackles a more difficult problem, where the model is only considered accurate if the exact category is estimated. For \(p = 4\), median accuracy of M-B-ECM Cat is lower than that of C-ECM. However, relative accuracy improves for increasing \(p\), and M-B-ECM Cat clearly performs better than C-ECM for \(p = 8\) and 10.  This result shows that the utility of using partial observations is quite high because M-B-ECM Cat is tackling a more difficult problem than C-ECM.

Table \ref{tab:syn-data} provides results on the accuracy rate, false negative rate, and false positive rate calculated over the entirety of the experiment.  Values shown in parentheses indicate the benefit or reduction in benefit from using a B-ECM model over C-ECM in this specific problem, with red characters indicating a strongly undesirable change, orange characters indicating a slightly undesirable change, yellow characters indicating little change, green characters indicating a slight improvement, and blue characters indicating a strong improvement.  C-ECM has a relatively high false negative rate and a low false positive rate. B-ECM trained on the same data provides slight improvements in accuracy, with the magnitude of improvement increasing for increasing values of \(p\). B-ECM significantly improves upon the false negative rate of C-ECM, and has a slightly worse false positive rate.  

M-B-ECM has much larger improvements in accuracy, especially when \(p = 10\).  For \(p = 10\) the false positive rate of M-B-ECM is similar to that of C-ECM and the false negative rate is much improved over C-ECM.  Table \ref{tab:syn-data} shows that using the same M-B-ECM fit, but changing the loss function to target a reduction in false negatives, M-B-ECM \(C_{1,2} = 2\) trades a small reduction in accuracy and increase in the false positive rate for a further decrease in the false negative rate over M-B-ECM.  For \(p = 4\), M-B-ECM \(C_{1,2} = 2\) has a significantly lower false negative rate.  For smaller values of \(p\) and constant \(N\), the problem is more difficult for all models.  Changing the loss function to reduce false negatives allows us to be more cautious given the increased uncertainty, with little penalty to overall accuracy.

For full categorisation using M-B-ECM Cat, the threshold for the value of \(p(\tilde{z}_k = 1|\tilde{\bm{y}}_{\tilde{p}}, \bm{Y}_{N\times p}, \bm{\theta})\) to categorize \(\tilde{\bm{y}}_{\tilde{p}}\) in the category of interest indexed as \(k\) is lower than that of binary categorization for M-B-ECM.  Naturally, this reduction in the threshold reduces the false negative rate while raising the false positive rate when compared to M-B-ECM using 0-1 loss.  For \(p = 10\), values for all performance metrics are similar for both M-B-ECM and M-B-ECM Cat.  The relative performance improvements of these models over C-ECM illustrates the flexibility of B-ECM for adapting to data fusion applications with larger number of discrimianants and covariance matrices with no inter-category dependence.
%Indeterminant and undefined results from a C-ECM model would be subject to further review in practice.  Here such results count against accuracy and increase the false negative rate.  We care about the performance of an automatable model in this experiment, used to reduce the mass of constantly streaming data which needs to be examined by humans.  Sidelining an observation for further human review is in essence a failure of the model, so we feel that categorizing such instances as an inaccuracy is fair.  

%% Redid the tables with 5 colors instead of the previous 256.  Felt like 5 colors looked better and was more informative than 3.
\begin{table}[!h]
\caption{\label{tab:syn-data}Empirical results from synthetic data experiments.  Values in parenthesis are the difference observed from C-ECM, and are colored according to the magnitude of the difference as well as the benefit observed by utilizing the specific B-ECM implementation over C-ECM.}
\centering
\begin{tabular}[t]{lrrrr}
\toprule
  & $p$ = 4 & $p$ = 6 & $p$ = 8 & $p$ = 10\\
\midrule
\addlinespace[0.3em]
\multicolumn{5}{l}{\textbf{C-ECM}}\\
\hspace{1em}\cellcolor{gray!10}{Accuracy} & \cellcolor{gray!10}{0.73} & \cellcolor{gray!10}{0.75} & \cellcolor{gray!10}{0.76} & \cellcolor{gray!10}{0.75}\\
\hspace{1em}False Negative & 0.74 & 0.67 & 0.64 & 0.69\\
\hspace{1em}\cellcolor{gray!10}{False Positive} & \cellcolor{gray!10}{0.04} & \cellcolor{gray!10}{0.04} & \cellcolor{gray!10}{0.04} & \cellcolor{gray!10}{0.03}\\
\addlinespace[0.3em]
\multicolumn{5}{l}{\textbf{B-ECM}}\\
\hspace{1em}Accuracy & 0.77 (\textcolor[HTML]{EACB2B}{$\Delta$0.04}) & 0.82 (\textcolor[HTML]{8BBD94}{$\Delta$0.07}) & 0.86 (\textcolor[HTML]{8BBD94}{$\Delta$0.1}) & 0.9 (\textcolor[HTML]{3B99B1}{$\Delta$0.15})\\
\hspace{1em}\cellcolor{gray!10}{False Negative} & \cellcolor{gray!10}{0.46 (\textcolor[HTML]{6CB799}{$\Delta$-0.28})} & \cellcolor{gray!10}{0.35 (\textcolor[HTML]{6CB799}{$\Delta$-0.32})} & \cellcolor{gray!10}{0.25 (\textcolor[HTML]{6CB799}{$\Delta$-0.39})} & \cellcolor{gray!10}{0.18 (\textcolor[HTML]{3B99B1}{$\Delta$-0.51})}\\
\hspace{1em}False Positive & 0.11 (\textcolor[HTML]{E79812}{$\Delta$0.08}) & 0.1 (\textcolor[HTML]{E79812}{$\Delta$0.06}) & 0.09 (\textcolor[HTML]{EACB2B}{$\Delta$0.04}) & 0.06 (\textcolor[HTML]{EACB2B}{$\Delta$0.03})\\
\addlinespace[0.3em]
\multicolumn{5}{l}{\textbf{M-B-ECM}}\\
\hspace{1em}\cellcolor{gray!10}{Accuracy} & \cellcolor{gray!10}{0.79 (\textcolor[HTML]{8BBD94}{$\Delta$0.06})} & \cellcolor{gray!10}{0.85 (\textcolor[HTML]{8BBD94}{$\Delta$0.1})} & \cellcolor{gray!10}{0.89 (\textcolor[HTML]{3B99B1}{$\Delta$0.13})} & \cellcolor{gray!10}{0.92 (\textcolor[HTML]{3B99B1}{$\Delta$0.17})}\\
\hspace{1em}False Negative & 0.45 (\textcolor[HTML]{6CB799}{$\Delta$-0.29}) & 0.31 (\textcolor[HTML]{6CB799}{$\Delta$-0.36}) & 0.21 (\textcolor[HTML]{6CB799}{$\Delta$-0.43}) & 0.15 (\textcolor[HTML]{3B99B1}{$\Delta$-0.53})\\
\hspace{1em}\cellcolor{gray!10}{False Positive} & \cellcolor{gray!10}{0.08 (\textcolor[HTML]{EACB2B}{$\Delta$0.05})} & \cellcolor{gray!10}{0.08 (\textcolor[HTML]{EACB2B}{$\Delta$0.03})} & \cellcolor{gray!10}{0.06 (\textcolor[HTML]{EACB2B}{$\Delta$0.02})} & \cellcolor{gray!10}{0.04 (\textcolor[HTML]{EACB2B}{$\Delta$0.01})}\\
\addlinespace[0.3em]
\multicolumn{5}{l}{\textbf{M-B-ECM \(C_{1,2} = 2\)}}\\
\hspace{1em}Accuracy & 0.76 (\textcolor[HTML]{EACB2B}{$\Delta$0.03}) & 0.83 (\textcolor[HTML]{8BBD94}{$\Delta$0.08}) & 0.87 (\textcolor[HTML]{8BBD94}{$\Delta$0.12}) & 0.92 (\textcolor[HTML]{3B99B1}{$\Delta$0.17})\\
\hspace{1em}\cellcolor{gray!10}{False Negative} & \cellcolor{gray!10}{0.23 (\textcolor[HTML]{3B99B1}{$\Delta$-0.51})} & \cellcolor{gray!10}{0.17 (\textcolor[HTML]{3B99B1}{$\Delta$-0.5})} & \cellcolor{gray!10}{0.14 (\textcolor[HTML]{3B99B1}{$\Delta$-0.5})} & \cellcolor{gray!10}{0.11 (\textcolor[HTML]{3B99B1}{$\Delta$-0.58})}\\
\hspace{1em}False Positive & 0.24 (\textcolor[HTML]{F5191C}{$\Delta$0.21}) & 0.17 (\textcolor[HTML]{E79812}{$\Delta$0.13}) & 0.12 (\textcolor[HTML]{E79812}{$\Delta$0.08}) & 0.07 (\textcolor[HTML]{EACB2B}{$\Delta$0.04})\\
\addlinespace[0.3em]
\multicolumn{5}{l}{\textbf{M-B-ECM Cat}}\\
\hspace{1em}\cellcolor{gray!10}{Accuracy} & \cellcolor{gray!10}{0.67 (\textcolor[HTML]{E78500}{$\Delta$-0.06})} & \cellcolor{gray!10}{0.76 (\textcolor[HTML]{EACB2B}{$\Delta$0.01})} & \cellcolor{gray!10}{0.83 (\textcolor[HTML]{8BBD94}{$\Delta$0.07})} & \cellcolor{gray!10}{0.89 (\textcolor[HTML]{3B99B1}{$\Delta$0.14})}\\
\hspace{1em}False Negative & 0.35 (\textcolor[HTML]{6CB799}{$\Delta$-0.39}) & 0.26 (\textcolor[HTML]{6CB799}{$\Delta$-0.41}) & 0.19 (\textcolor[HTML]{3B99B1}{$\Delta$-0.45}) & 0.14 (\textcolor[HTML]{3B99B1}{$\Delta$-0.54})\\
\hspace{1em}\cellcolor{gray!10}{False Positive} & \cellcolor{gray!10}{0.15 (\textcolor[HTML]{E79812}{$\Delta$0.12})} & \cellcolor{gray!10}{0.11 (\textcolor[HTML]{E79812}{$\Delta$0.07})} & \cellcolor{gray!10}{0.08 (\textcolor[HTML]{EACB2B}{$\Delta$0.04})} & \cellcolor{gray!10}{0.05 (\textcolor[HTML]{EACB2B}{$\Delta$0.02})}\\
\bottomrule
\end{tabular}
\end{table}

\begin{figure}
\centering
\includegraphics{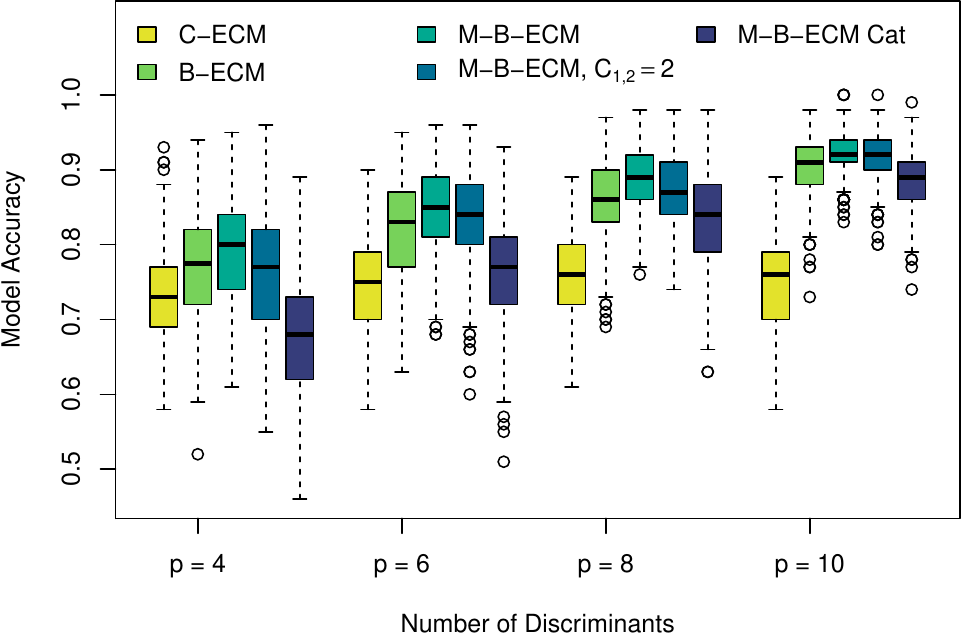}
\caption{Box plots showing the distribution of observed accuracy over the MC iterations for the synthetic data generating mechanism.}
\label{fig:syn-boxplot}
\end{figure}

\subsection{Seismic Discriminant Data}\label{sec:dale-exp}

We now implement our competing versions of ECM against real observations that researchers collected from ground-deployed sensors (seismometers) \citep{Anderson2008_1}. Observers computed the discriminants in this dataset entirely from seismic waveforms that were sourced by real events that include deep earthquakes, shallow earthquakes, and various explosions \ref{fig:ecm-map}. The size of these events (seismic magnitude or explosion yield) largely determined whether a given distribution of seismometers could record signals over the ambient noise environment well enough to estimate source type discriminants. 

The dataset omits discriminants in numerous cases that signals are absent from a sufficient number of seismometer observations; therefore, the dataset includes missing entries. One of the most common discriminants that is present in this dataset is the so-called signal complexity factor, $\beta_{CF}$ (see Fig. \ref{fig:matrix-format}). This factor measures the log-ratio $\log\left( \frac{E_c}{E_s}\right)$ between seismic waveform coda energy $E_c$ and seismic signal energy $E_s$. Here, the coda wave energy $E_c$ is measured over a window that begins 5 seconds after the first compressional wave arrival and that ends 25 seconds after its arrival. The signal energy $E_s$ is measured in a 5 second window that begins immediately after the compressional wave arrival. Researchers within the United Kingdom (UK) Atomic Weapons Establishment (AWE) formed these measurements from seismic array beams in the UK (station code EKA), India (station code GBA), Australia (station code WRA) and Canada (station code YKA) that they filtered over a passband of 0.25 to 4 Hz. Observations demonstrate that waveforms sourced by both nuclear explosions and deep earthquakes show relative simple waveforms and less scattered energy; these produce negative value of $\beta_{CF}$. The other discriminants present in our AWE dataset are less populous but more conventional; they include earthquake depth estimates and body wave versus surface wave magnitudes. We refer readers to \citet{Anderson2008_1} for a summary of their mathematical forms.

The categorization problem that we consider thereby includes only three event type categories: explosions, deep earthquakes, and shallow earthquakes. The dataset groups nuclear and conventional explosion events together, since seismic data cannot generally discriminate between explosion source type (although body wave magnitudes from conventional explosions are usually less than those of nuclear explosions) . We then reduce the total number of available discriminants to a subset that includes a sufficient number (five) to train the C-ECM model. The resulting data set thereby contains five discriminants computed from 280 observations composed of 155 explosions, 26 deep earthquakes, and 99 shallow earthquakes. Our reduction in discriminants still retains a dataset that has 54\% of the \(280 \times 5 = 1{,}400\) of its elements missing.  A small number of observations (25) contain data for all five discriminants; 12 explosions, two deep earthquakes, and 11 shallow earthquakes. We were unable to collect a combination of \(p > 5\) that increased the number of full observations. We therefore used \(p = 5\) within the MC experiment.

\begin{figure}
\centering
\includegraphics{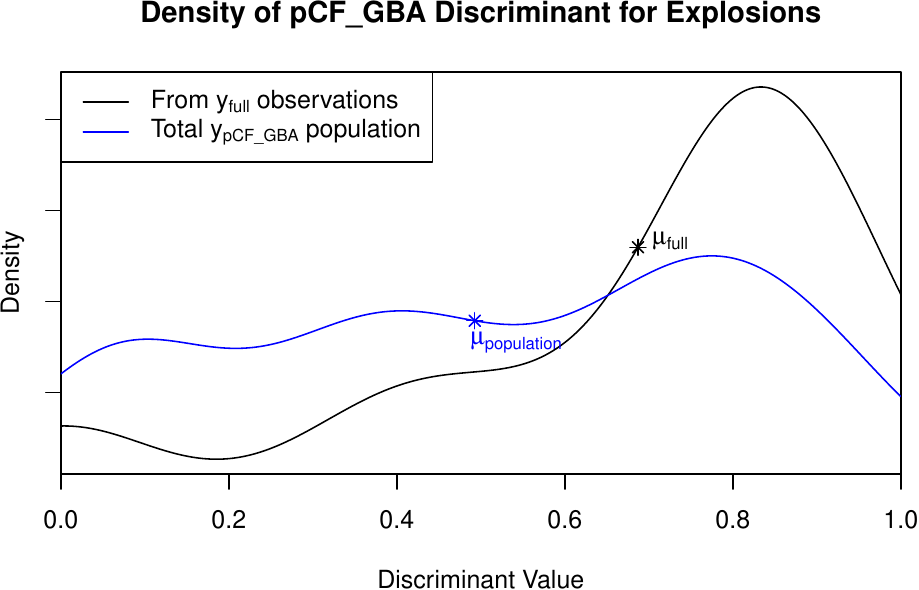}
\caption{Density fits of subsets of the pCF\_GBA discriminant recorded from explosions.  The density fits are from the events for which the full \(p = 5\) discriminants were recorded (shown in black) and the total population of the pCF\_GBA data discriminant for explosions (shown in blue).  The mean values, noted as \(\mu\), were taken after the logit transform and transformed back to \([0,1]\) for plotting.}
\label{fig:differing-densities}
\end{figure}

Missing entries within the dataset do not appear to be missing at random. Fig. \ref{fig:differing-densities} illustrates this property with p-values computed from the signal complexity factors computed from explosions at the GBA array (written as pCF\_GBA). Density fits of this discriminant from the total population and the subset of the data that is associated with full \(p = 5\) observations show that, for this case, the mean of the data taken from full \(p = 5\) observations is shifted significantly from the population mean. This implies that the distributions of values obtained from events with different degrees of data ``missingness'' may not be the same.

The experiment consisted of 250 MC iterations. Subsets of the data were sampled without replacement for training and testing, for each iteration. Because there were only two deep earthquake observations with full observations, and because the functions used for C-ECM require at least two full observations for each category, both of these observations were included in every training data set.  The 18 additional full observations were randomly sampled from the data set to include in the training data.  Of the remaining 255 partial observations, we used \(\left \lfloor{255 \times \sqrt{p}/(1 + \sqrt{p})}\right \rfloor = 176 \) \citep{joseph2022optimal} for training. We used the remaining 84 observations, 79 partial observations and 5 full observations for testing.

The distribution of the observed model accuracy over each MC iteration is shown as box plots in Fig. \ref{fig:dale-boxplot}.  There is relatively less variability in the results from this real data set compared to Fig. \ref{fig:syn-boxplot}.  Median accuracy of M-B-ECM, M-B-ECM \(C_{1,2} = 2\), and M-B-ECM Cat is clearly better than that of C-ECM and B-ECM.  If discriminants are not missing at random, the methods which utilize partial observations for training can get a better representation of the mean of the population, useful for predicting the category of \(\tilde{\bm{y}}_{\tilde{p}}\) for testing data where \(\tilde{p} \in \{1, \dots , 5\}\).  

The results of the total accuracy, false negative rate, and false positive rate calculated over the entirety of data collected are shown in Table \ref{tab:dale-data}.  Results largely echo those in Table \ref{tab:syn-data}.  B-ECM has slightly higher overall accuracy than C-ECM for this moderate \(p = 5\) problem.  C-ECM has a much higher false negative rate than all B-ECM comparators and a lower false positive rate.  

Including missing data to train a M-B-ECM model generally resulted in improved accuracy, a greatly reduced false negative rate, and a worse false positive rate.  For many of the explosion discriminants, the population has a higher variance than the portion which is part of a full \(\tilde{p} = 5\) event.  We hypothesize that the increase in variance for these discriminants within the training data translated to more true explosion \(\tilde{\bm{y}}_{\tilde{p}}\) being captured as being more probable.  The deep earthquake and shallow earthquake discriminant populations often have a lower variance than their explosion counterparts.  Conversely, this increase in variance could also lead to more \(\tilde{\bm{y}}_{\tilde{p}}\) which are not truly explosions to be categorized as such.

M-B-ECM \(C_{1,2} = 2\) resulted in a further reduction in the false negative rate, with little trade off in overall accuracy.  Similar to what can be noted from Table \ref{tab:syn-data}, the lower threshold for categorization as an explosion results in a lower false negative rate and higher false positive rate than binary categorization with M-B-ECM and 0-1 loss.

\begin{table}
\caption{\label{tab:dale-data}Empirical results from seismic discriminant experiments.  Values in parenthesis are the difference observed from C-ECM, and are colored according to the magnitude of the difference as well as the benefit observed by utilizing the specific B-ECM implementation over C-ECM.}
\centering
\begin{tabular}[t]{lr}
\toprule
  & $p$ = 5\\
\midrule
\addlinespace[0.3em]
\multicolumn{2}{l}{\textbf{C-ECM}}\\
\hspace{1em}\cellcolor{gray!6}{Accuracy} & \cellcolor{gray!6}{0.71}\\
\hspace{1em}False Negative & 0.51\\
\hspace{1em}\cellcolor{gray!6}{False Positive} & \cellcolor{gray!6}{0.01}\\
\addlinespace[0.3em]
\multicolumn{2}{l}{\textbf{B-ECM}}\\
\hspace{1em}Accuracy & 0.75 (\textcolor[HTML]{8BBD94}{$\Delta$0.04})\\
\hspace{1em}\cellcolor{gray!6}{False Negative} & \cellcolor{gray!6}{0.41 (\textcolor[HTML]{EACB2B}{$\Delta$-0.09})}\\
\hspace{1em}False Positive & 0.03 (\textcolor[HTML]{EACB2B}{$\Delta$0.02})\\
\addlinespace[0.3em]
\multicolumn{2}{l}{\textbf{M-B-ECM}}\\
\hspace{1em}\cellcolor{gray!6}{Accuracy} & \cellcolor{gray!6}{0.85 \vphantom{1} (\textcolor[HTML]{3B99B1}{$\Delta$0.14})}\\
\hspace{1em}False Negative & 0.12 (\textcolor[HTML]{3B99B1}{$\Delta$-0.39})\\
\hspace{1em}\cellcolor{gray!6}{False Positive} & \cellcolor{gray!6}{0.18 (\textcolor[HTML]{F5191C}{$\Delta$0.17})}\\
\addlinespace[0.3em]
\multicolumn{2}{l}{\textbf{M-B-ECM} \(C_{1,2} = 2\)}\\
\hspace{1em}Accuracy & 0.85 (\textcolor[HTML]{3B99B1}{$\Delta$0.14})\\
\hspace{1em}\cellcolor{gray!6}{False Negative} & \cellcolor{gray!6}{0.09 (\textcolor[HTML]{3B99B1}{$\Delta$-0.42})}\\
\hspace{1em}False Positive & 0.24 (\textcolor[HTML]{F5191C}{$\Delta$0.22})\\
\addlinespace[0.3em]
\multicolumn{2}{l}{\textbf{M-B-ECM Cat}}\\
\hspace{1em}\cellcolor{gray!6}{Accuracy} & \cellcolor{gray!6}{0.79 (\textcolor[HTML]{8BBD94}{$\Delta$0.08})}\\
\hspace{1em}False Negative & 0.1 (\textcolor[HTML]{3B99B1}{$\Delta$-0.41})\\
\hspace{1em}\cellcolor{gray!6}{False Positive} & \cellcolor{gray!6}{0.21 (\textcolor[HTML]{F5191C}{$\Delta$0.2})}\\
\bottomrule
\end{tabular}
\end{table}

\begin{figure}
\centering
\includegraphics{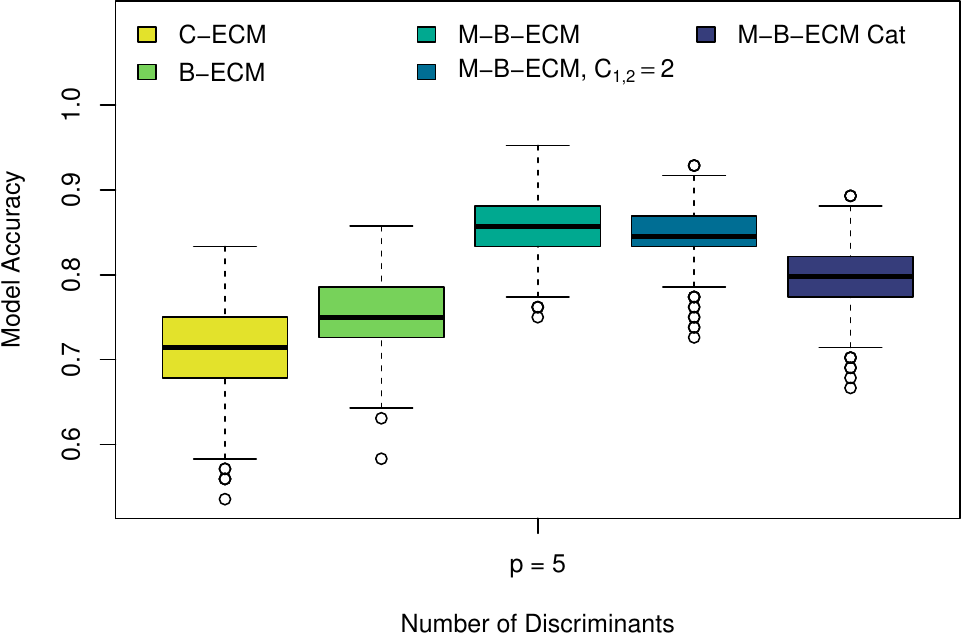}
\caption{Box plots showing the distribution of observed accuracy over the MC iterations for the seismic discriminant data set.}
\label{fig:dale-boxplot}
\end{figure}

\section{Discussion and Conclusions}\label{sec:discussion}

Results consistently suggest that a B-ECM model trained on observations with no missing data has accuracy similar to or greater than a C-ECM model, as well as a lower false negative rate.  The inclusion of additional training events where some data is missing for a B-ECM model consistently results in clear improvements in categorization accuracy and reductions in false negative rate over C-ECM.  These findings indicate that B-ECM is a worthy competitor for event identification in explosion monitoring.  The experiments in Sections \ref{sec:synthetic-exp} and \ref{sec:dale-exp} highlight these differences in performance.  While C-ECM has relatively low false positive rates for all data used and all sizes of \(p\) tested, B-ECM typically performs better in all other aspects for these difficult problems, and has a similar false positive rate for some problems.  Adjusting the values of the loss matrix for applications where a lower false positive rate is desirable facilitates such a reduction if desired.

A B-ECM model which can handle training data with missing elements, assumes the data is missing completely at random.  A combination of intuition and the evidence displayed in Fig. \ref{fig:differing-densities} leads us to believe missing data elements are not truly random.  An intuitive hypothetical would be a low yield weapon not meeting certain thresholds in order to start automated data recording, resulting in some missing data for such an event.  However, even with this nuanced model misspecification, simply having the ability to utilize data with missing entries results in significant performance gains.  We expect further improvements in performance for a B-ECM model which instead has a missing not at random \citep{little2019statistical} specification.

The decision theoretic framework for B-ECM provides intuition to tune ``knobs'', by changing values within the loss matrix, as in Equation (\ref{eq:loss-false-pos}).  Values within the loss matrix could be chosen subjectively, as is done in Section \ref{sec:exp}, or tuned empirically given a large enough data set, in order to target a reduction in false negatives or false positives.  When we increase the value of element \(C_{1,2}\), logically we are increasing the loss associated with erroneously choosing to categorize \(\tilde{\bm{y}}_{\tilde{p}}\) as not a detonation when the event truly is a detonation.  The intuitive relationship between the values of the elements of a loss matrix has utility in operations.  Just as important, changing the values of the loss matrix impacts the results as intended.  In our experiments, increasing \(C_{1,2}\) to a value of 2 in order to reduce the false negative rate did result in the intended reduction by roughly a factor of two, with only slight reductions in overall accuracy.

The B-ECM model is flexible enough to adapt to both detection and categorization applications.  Sets of elements from the \(K\) length vector from the predictive category distribution \(p(\tilde{\bm{z}}_K|\tilde{\bm{y}}_{\tilde{p}}, \bm{Y}_{N \times p}, \bm{\theta})\) can be summed to group event categories together, so that B-ECM can be used to categorize \(\tilde{\bm{y}}_{\tilde{p}}\) into anywhere from 2 to \(K\) groups with a corresponding adjustment of the loss matrix.  A larger number of actions to choose from corresponds with a decrease in accuracy.  This increase in the complexity of the problem results in the observed decrease in accuracy for both data sets explored in experimentation.  However, this effect appears to diminish as \(p\) increases, which we hypothesize is due to a decrease in the overlap of \(p(\tilde{\bm{y}}_{\tilde{p}}| \tilde{z}_k = 1, \bm{Y}_{N_k \times p}, \bm{\theta}_k)\) in higher dimension for all \(k\)  Our testing with event categorization utilizes data with missing entries.  Even though the categorization problem is more difficult, at times event categorization had higher accuracy than methods which did not take advantage of missing data.  This was particularly true in Section \ref{sec:dale-exp}.  These results illustrate how powerful utilizing all data can be for more difficult problems.

We introduced a novel decision framework for the Bayesian typicality index, which is able to detect outliers under the uncertainty imparted by using partial data observations.  The use of the Bayesian typicality index ensures that a \(\tilde{\bm{y}}_{\tilde{p}}\) which was not generated from one of the \(K\) events has the possibility of being categorized as an outlier.  Without the typicality index, a new observation would be required to be categorized under the finite set of event categories used for training.  The Bayesian typicality index is a B-ECM model specific twist on established methodology, which can take into account uncertainty related to using training data with missing entries. 

With the ability to handle partial observations and the use of Bayesian Decision theory, B-ECM has been shown to be an improvement over C-ECM, and there are many avenues for continuing improvement even further over the current state of the art.  Values of the matrix \(\bm{C}\) and typicality index significance level \(\tilde{\alpha}\) can be tuned to meet an application specific objective.  The utility of loss matrices where the elements are functions instead of constants can be investigated \citep{robert2007bayesian, berger2013statistical}.  We believe taking advantage of Bayesian model selection procedures would produce further improvements in accuracy, especially in applications where data is fused over multiple modalities.  In a preliminary investigation, we did not see consistent improvements in accuracy for the the logit function over the transform used in \citet{anderson2007mathematical}.  The preferable transform depends on the data itself, and the data can inform which transform to use as part of a Bayesian model selection procedure using Bayes Factors \citep{kass1995bayes}.  Additionally, model selection procedures can be used to allow uniquely sparse covariance matrices \citep{jordan1999learning, roverato2002hyper} for each category and models which assume missing data is not at random \citep{little2019statistical}.  For event categories where there is no correlation between particular discriminants, such a model may be a better fit to the data.

\section*{Author Contribution Statement}

\begin{itemize}
\item Scott Koermer: Conceptualization, Formal analysis, Investigation, Methodology, Software, Validation, Visualization, Writing - original draft, Writing - review \& editing
\item Joshua D. Carmichael: Data curation, Supervision, Methodology, Visualization, Writing - original draft, Writing - review \& editing
\item Brian J. Williams:  Methodology, Supervision, Writing - original draft, Writing - review \& editing
\end{itemize}

\section*{Data Availability}

Code and data to generate the results can be found within the \texttt{ezECM} package for the \textsf{R} programming language.  Further details are provided in Appendices \ref{sec:synthetic-data-code} and \ref{sec:data}.

%\section*{Funding}
\section*{Acknowledgement}

This manuscript has been authored with number LA-UR-24-30189 by Triad National Security under Contract with the U.S. Department of Energy, Office of Defense Nuclear Nonproliferation Research and Development. This research was funded by the National Nuclear Security Administration, Defense Nuclear Nonproliferation Research and Development (NNSA DNN R\&D). The authors acknowledge important interdisciplinary collaboration with scientists and engineers from LANL, LLNL, MSTS, PNNL, and SNL. Los Alamos National Laboratory is supported by the U.S. Department of Energy National Nuclear Security Administration under Contract No. DE-AC52-06NA25396. The United States Government retains and the publisher, by accepting the article for publication, acknowledges that the United States Government retains a non-exclusive, paid-up, irrevocable, world-wide license to publish or reproduce the published form of this manuscript, or allow others to do so, for United States Government purposes.

Research presented in this article was supported by the Laboratory Directed Research and Development program of Los Alamos National Laboratory under project number 20220188DR. Los Alamos National Laboratory is operated by Triad National Security, LLC, for the National Nuclear Security Administration of U.S. Department of Energy (Contract No. 89233218CNA000001).

\printbibliography

\appendix

%\section{Nomenclature}
%https://www.overleaf.com/learn/latex/Nomenclatures

%% Model params B

\nomenclature[B]{\(\bm{\mu}_k\)}{Mean vector of length \(p\) a multivariate normal distribution for the \(k^{\mathrm{th}}\) event category.}
\nomenclature[B]{\(\bm{\Sigma}_k\)}{\(p \times p\) covariance matrix from the multivariate normal distribution of the \(k^{\mathrm{th}}\) event category.  Parameterizes the covariance between the columns of \(\bm{Y}_{N_k \times p}\)}
\nomenclature[B]{\(\bm{\pi}\)}{Vector of weights for the mixture of distributions.  \(\bm{1}_K^{\top}\bm{\pi} = 1\)}
\nomenclature[B]{\(\pi_k\)}{Weight of the \(k^{\mathrm{th}}\) event category distribution within the mixture of distributions.}
\nomenclature[B]{\(\tilde{\bm{z}}_K\)}{Vector of latent variables of length \(K\) used to model the event category of \(\tilde{\bm{y}}_{\tilde{p}}\).  A single \(k^{\mathrm{th}}\) element is equal to one, indicating inclusion into the \(k^{\mathrm{th}}\) group, with the remaining \(K -1\) elements equal to zero.}
\nomenclature[B]{\(\tilde{z}_k\)}{The \(k^{\mathrm{th}}\) element of \(\tilde{\bm{z}}_K\)}
\nomenclature[B]{\(\mathbb{E}[\tilde{\bm{z}}_K|\tilde{\bm{y}}_{\tilde{p}}, \bm{Y}_{N\times p}, \bm{\theta}]\)}{Vector of probability of each category given a new observation, a set of training data, and prior hyperparameters.  The posterior predictive expected value of each \(\tilde{z}_k\).}

%% Prior hyperparams P

\nomenclature[P]{\(\bm{\theta}_k\)}{Short hand for the group hyperparameters \(\{\bm{\eta}_k, \bm{\Psi}_k, \nu_k, \bm{\alpha}\}\) related to the prior distributions of parameters \(\{\bm{\mu}_k, \bm{\Sigma}_k, \pi_k\}\)}
\nomenclature[P]{\(\bm{\theta}\)}{The full collection of prior hyperparameters \(\{\bm{\theta}_1, \dots , \bm{\theta}_K\}\)}
\nomenclature[P]{\(\bm{\eta}_k\)}{Prior mean of \(\bm{\mu}_k\)}
\nomenclature[P]{\(\bm{\Psi}_k\)}{Prior covariance of \(\bm{\mu}_k\), and prior scale matrix for \(\bm{\Sigma}_k\)}
\nomenclature[P]{\(\nu_k\)}{Prior degrees of freedom for \(\bm{\Sigma}_k\)}
\nomenclature[P]{\(\bm{\alpha}\)}{Vector of prior concentration parameters for \(\bm{\pi}\)}
\nomenclature[P]{\(\alpha_k\)}{Prior concentration parameter for \(\pi_k\)}
\nomenclature[P]{\(\bm{\nu}\)}{Collection of all \(K\) values of \(\nu_k\).}
\nomenclature[P]{\(\bm{\eta}\)}{Collection of all \(K\) values of \(\bm{\eta}\).}
\nomenclature[P]{\(\bm{\Psi}\)}{Collection of all \(K\) values of \(\bm{\Psi}_k\).}

%% Algorithm A

\nomenclature[A]{\(T\)}{Total number of iterative samples taken of a single random variable in Markov chain Monte Carlo.}
\nomenclature[A]{\(B\)}{Number of burn-in samples of a random variable discarded immediately after starting the Markov chain.}
\nomenclature[A]{\((x)^{(t)}\)}{The \(t^{\mathrm{th}}\) iterative sample of random variable \(x\) obtained in a Markov chain Monte Carlo algorithm.}
\nomenclature[A]{\(\bm{P}_{k, \ell}^{R}\)}{Matrix that permutes rows containing the missing entries in the \(\ell^{\mathrm{th}}\) column of \(\bm{Y}_{N_k \times p}\) to the top \(N_k^m\) rows.}
\nomenclature[A]{\(\bm{P}_{k, \ell}^{C}\)}{Matrix that permutes the \(\ell^{\mathrm{th}}\) column of \(\bm{Y}_{N_k \times p}\) to the first column.}
\nomenclature[A]{\(\underline{\bm{X}}\)}{Permuted matrix or vector, for arbitrary matrix or vector \(\bm{X}\), re-arranged for conditional sampling.}
\nomenclature[A]{\(\underline{\bm{y}}_{N_k^m \times 1}^{-}\)}{Vector subset of the missing entries within the first column of \(\underline{\bm{Y}}_{N_k \times p}\).}
\nomenclature[A]{\(\underline{\bm{y}}_{N_k^o \times 1}^{+}\)}{Vector subset of the observed entries within the first column of \(\underline{\bm{Y}}_{N_k \times p}\).}
\nomenclature[A]{\(\bmk{Z}_{K \times (T-B)}\)}{Matrix used for storing draws of the random variable \(\mathbb{E}[\tilde{\bm{z}}_K|\tilde{\bm{y}}_{\tilde{p}}, \bm{Y}_{N\times p}, \bm{\theta}]\).}
\nomenclature[A]{\(\underline{\bm{Y}}_{N_k^m \times (p -1)}\)}{Subset of the first \(N_k^m\) rows of \(\underline{\bm{Y}}_{N_k \times p}\), with the first column removed.}
\nomenclature[A]{\(\underline{\bm{Y}}_{N_k^o \times (p -1)}\)}{Subset of rows \(N_k^m + 1\) through \(N_k\) of \(\underline{\bm{Y}}_{N_k \times p}\), with the first column removed.}
\nomenclature[A]{\(q\)}{Integer parameter for thinning samples from the Markov Chain, where after burn-in only every \(q^{\mathrm{th}}\) sample is saved.}
\nomenclature[A]{\(p_k\)}{Value proportional to the probability of \(\tilde{\bm{y}}_{\tilde{p}}\) belonging to category \(k\) before normalization.  Equivalent to \(p(\tilde{\bm{y}}_{\tilde{p}}|\tilde{z}_k = 1, \bm{Y}_{N_k\times p}^{(t)}, \bm{\theta}_k)p(\tilde{z}_k = 1|\bm{Y}_{N \times p},\bm{\theta})\).}
\nomenclature[A]{\(p_K\)}{\(\sum_k = 1^K p_k\), the normalizing constant used to evaluate \(p(\tilde{z}_k = 1|\tilde{\bm{y}}_{\tilde{p}}, \bm{Y}_{N \times p}, \bm{\theta})\).}
\nomenclature[A]{\(\bm{y}_{k, \ell}\)}{Missing data elements of column \(\ell\) in \(\bm{Y}_{N_k \times p}\).}
\nomenclature[A]{\(\underline{\bm{Y}}_{N_k \times p}\)}{Permuted matrix of training data for the \(k^{\mathrm{th}}\) event category.  Equal to \(\bm{P}^{R}_{k, \ell} \bm{Y}_{N_k \times p} \bm{P}^{C}_{k,\ell}\) for a given column \(\ell\).}
\nomenclature[A]{\(t\)}{Indexes the \(T\) Monte Carlo samples of a random variable.}
\nomenclature[A]{\(q\)}{Integer thinning factor used to specify the retention of every \(q^{\mathrm{th}}\) sample in Markov chain Monte Carlo in an effort to reduce auto-correlation between what would ideally be independent samples.}
\nomenclature[A]{\((\bm{Y}_{N_k \times p})^{(t)}\)}{The union of the observed data \(\bm{Y}_k^{+}\) and the \(t^{\mathrm{th}}\) Markov chain Monte Carlo random sample of missing data \((\bm{Y}_k^{-})^{(t)}\) for the \(k^{\mathrm{th}}\) event category.}
%\nomenclature[A]{\(\mathrm{dim}(\bm{x})\)}{Dimension of vector \(\bm{x}\)}
\nomenclature[A]{\(\bm{X}_{[\bm{a}, \bm{b}]}\)}{Subset arbitrary matrix \(\bm{X}\) as indicated by vectors \(\bm{a}\) and \(\bm{b}\)}
%\nomenclature[A]{\texttt{NA}}{A placeholder for a missing data element.}

%% Decision Theory T

\nomenclature[T]{\(a_k\)}{The \(k^{\mathrm{th}}\) categorization action.}
\nomenclature[T]{\(\bm{a}\)}{The set of actions \(\{a_1, a_2\}\) for binary categorization or \(\{a_1 , \dots , a_K\}\) for full \(K\) categorization.}
\nomenclature[T]{\(\bm{C}\)}{The loss matrix specifying the loss associated with each \(a_k\) along the columns and the true value of \(\tilde{z}_1\) or \(\tilde{\bm{z}}_K\) along the rows.  Of dimension \(2 \times 2\) for binary categorization and \(K\times K\) for categorization into the specific categories used in training.}
\nomenclature[T]{\(\bm{L}(\tilde{\bm{z}}_K, \bm{a})\)}{Vector loss evaluated for each action considered.  Vector values of length equal to the number of elements in \(\bm{a}\)}
\nomenclature[T]{\(\tilde{\alpha}\)}{Significance level used for hypothesis testing.}
\nomenclature[T]{\(C_{a,b}\)}{The element in the \(a^{\mathrm{th}}\) row and \(b^{\mathrm{th}}\) column of \(\bm{C}\).}
\nomenclature[T]{\(\delta(a_1)\)}{Indicator function for binary decisions, equal to 1 when \(a_1\) is the minimum expected loss action.}
\nomenclature[T]{\(Z\)}{Shorthand for the test statistic \((\tilde{\bm{y}}_{\tilde{p}} - \tilde{\bm{\mu}}_k)^{\top}\tilde{\bm{\Sigma}}_k^{-1}(\tilde{\bm{y}}_{\tilde{p}} - \tilde{\bm{\mu}}_k)/\tilde{p}\).}
\nomenclature[T]{\(\Phi\)}{A p-value calculated for the typicality index, with variability resulting from random draws of \(\bm{Y}_k^{-}\).  Conditioning on \(k, \bm{Y}_k^{+}, \bm{\theta}_k\) implied.}
\nomenclature[T]{\(p(\Phi)\)}{Distribution of random p-values, with variability resulting from random draws of \(\bm{Y}_k^{-}\).  Conditioning on \(k, \bm{Y}_k^{+}, \bm{\theta}\) implied.}
\nomenclature[T]{\(F^{\Phi}(\tilde{\alpha})\)}{Cumulative distribution function of \(\Phi\) evaluated at \(\tilde{\alpha}\).}

%% Discriminants D

\nomenclature[D]{\(\bm{Y}_{N \times p}\)}{Totality of the training data}
\nomenclature[D]{\(\bm{Y}_{N_k \times p}\)}{Observations used for training the \(k^{\mathrm{th}}\) event category.}
\nomenclature[D]{\(\bm{Y}_k^{-}\)}{The set of elements where data is missing from \(\bm{Y}_{N_k \times p}\)}
\nomenclature[D]{\(\bm{Y}_k^{+}\)}{The set of elements where data is not missing from \(\bm{Y}_{N_k \times p}\)}
\nomenclature[D]{\(\bm{Y}^{-}\)}{The set of missing elements from \(\bm{Y}_{N \times p}\).  Equivalent to \(\{\bm{Y}_{1}^{-}, \dots , \bm{Y}_{K}^{-}\}\)}
\nomenclature[D]{\(\bm{Y}^{+}\)}{The set of elements from \(\bm{Y}_{N \times p}\) which are not missing.  Equivalent to \(\{\bm{Y}_{1}^{+}, \dots , \bm{Y}_{K}^{+}\}\)}
\nomenclature[D]{\(\tilde{\bm{y}}_{\tilde{p}}\)}{Vector of observations from a new event with unknown category of length \(\tilde{p}\)}
\nomenclature[D]{\(\tilde{\bm{y}}_p\)}{Specifies a \(\tilde{\bm{y}}_{\tilde{p}}\) with no missing values, where \(\tilde{p} = p\).}
\nomenclature[D]{\({\bm{y}}_{N_k^m \times 1}^{-}\)}{Vector subset of the missing entries within the first column of \({\bm{Y}}_{N_k \times p}\).}
\nomenclature[D]{\({\bm{y}}_{N_k^o \times 1}^{+}\)}{Vector subset of the observed entries within the first column of \({\bm{Y}}_{N_k \times p}\).}
\nomenclature[D]{\({\bm{Y}}_{N_k^m \times (p -1)}\)}{Subset of the first \(N_k^m\) rows of \({\bm{Y}}_{N_k \times p}\), with the first column removed.}
\nomenclature[D]{\({\bm{Y}}_{N_k^o \times (p -1)}\)}{Subset of rows \(N_k^m + 1\) through \(N_k\) of \({\bm{Y}}_{N_k \times p}\), with the first column removed.}
\nomenclature[D]{\(\bar{\bm{y}}_k\)}{\(p \times 1 \) vector containing the mean of each column of \(\bm{Y}_{N_k \times p}\); \(\bar{\bm{y}}_k = \bm{Y}_{N_k \times p}^{\top} \bm{1}_{N_k} / N_k\)}

%% Mathematical Concepts S

\nomenclature[M]{\(\mathbb{E}[x]\)}{Expected value of random variable \(x\)}
\nomenclature[M]{\(\cup\)}{Union between two sets, herein often variables which do not have explicit dimensions.}
\nomenclature[M]{\(\mathbb{C}\mathrm{ov}(\bm{x})\)}{Covariance matrix between elements of random variable \(\bm{x}\).}
\nomenclature[M]{\(\left\lfloor{x}\right \rfloor\)}{The value of the random variable \(x\) rounded down to the nearest integer.}
\nomenclature[M]{\(\Gamma(x)\)}{The gamma function evaluated at \(x\).}
\nomenclature[M]{\(\Gamma_p(x)\)}{The multivariate gamma function evaluated at \(x\).}
\nomenclature[M]{\(\mathrm{vec}(\bm{X})\)}{Vectorization of matrix \(\bm{X}\), where columns of \(\bm{X}\) are sequentially arranged as the elements of a vector.}
\nomenclature[M]{\(\bm{A}\otimes \bm{B}\)}{Kronecker product of matrices \(\bm{A}\) and \(\bm{B}\).}
\nomenclature[M]{\(\mathbb{R}^p\)}{The space containing all \(p\) dimensional real numbers.}
\nomenclature[M]{\(\phi\)}{Statistical parameter to be inferred.  Used in notation for Bayes' rule.}
\nomenclature[M]{\(I_{x}(a,b)\)}{The regularized incomplete beta function.}
\nomenclature[M]{\(I[\cdot]\)}{Indicator function, equal to one if the expression in place of \(\cdot\) is true.}
\nomenclature[M]{\(\mathrm{logit}(x)\)}{Logit function on \(x\), equal to \(\ln(x) - \ln(1-x)\) where \(x \in (0,1)\).}
\nomenclature[M]{\(\mathrm{logit}^{-1}(x)\)}{Logistic function on \(x\), equal to \(1/(1 + \exp\{-x\})\) for \(x \in (-\infty , \infty)\).}
\nomenclature[M]{\(\forall\)}{For all elements, typically with respect to all elements in a set.}
\nomenclature[M]{\(\mathrm{tr}(\bm{X})\)}{Matrix trace of arbitrary square matrix \(\bm{X}\)}
\nomenclature[M]{\(x \gtrapprox y\)}{\(x\) greater than approximately \(y\)}
\nomenclature[M]{\(x \lessapprox y\)}{\(x\) less than approximately \(y\)}
\nomenclature[M]{\(\approx\)}{Approximately equal to}

%% Distributions N

\nomenclature[N]{\(\bm{t}_{p \times 1}\)}{Multivariate t-distributed random variable.}
\nomenclature[N]{\(t_{\nu}(\bm{\mu}, \bm{\Sigma})\)}{Multivariate t-distribution with \(\nu\) degrees of freedom, location parameter of \(\bm{\mu}\), and scale matrix of \(\bm{\Sigma}\).  Dimension can be defined according to the dimension of \(\bm{\mu}\) or \(\bm{\Sigma}\).}
\nomenclature[N]{\(\bm{T}_{N\times p}\)}{Matrix variate t-distributed random variable with \(N\) rows and \(p\) columns.}
\nomenclature[N]{\(T_{N \times p}(m, \bm{M}, \bm{\Sigma}, \bm{\Omega})\)}{Matrix variate t-distribution with degrees of freedom \(m\), mean matrix \(\bm{M}_{N \times p}\), row spread matrix \(\bm{\Sigma}_{N \times N}\), and column spread matrix \(\bm{\Omega}_{p \times p}\).}
\nomenclature[N]{\(\mathbb{F}(\tilde{p}, \tilde{\nu})\)}{F-distribution with degrees of freedom \(\tilde{p}\) and \(\tilde{\nu}\).}
\nomenclature[N]{\(\mathcal{N}(\bm{\mu}_p, \bm{\Sigma}_{p\times p})\)}{Multivariate normal distribution of dimension \(p\) with mean \(\bm{\mu}_p\) and covariance \(\bm{\Sigma}_{p \times p}\).}
\nomenclature[N]{\(\bm{M}_{N \times p}\)}{\(N \times p\) mean matrix for a \(N \times p\) matrix t-distributed random variable.}
\nomenclature[N]{\(\bm{\Omega}_{p \times p}\)}{\(p \times p\) column spread matrix for a \(N \times p\) matrix t-distributed random variable.}
\nomenclature[N]{\(\bm{\Sigma}\)}{Depending on context, either the covariance parameter of a multivariate normal distribution, scale matrix of a multivariate t-distribution, or row spread matrix of a matrix t-distribution.}
\nomenclature[N]{\(\bm{\mu}\)}{Mean vector of a multivariate normal distribution or multivariate t-distribution.}
\nomenclature[N]{\(\bm{\Lambda}_{11}, \bm{\Lambda}_{12}, \bm{\Lambda}_{21}, \bm{\Lambda}_{22}\)}{Short hand for linear combinations used for the conditional matrix t-distribution equations.}
\nomenclature[N]{\(\tilde{\bm{\mu}}_k\)}{Shorthand for the location vector of the multivariate t-distribution modeling the predictive density \(p(\tilde{\bm{y}}_{\tilde{p}}|\tilde{z}_k = 1, \bm{Y}_{N_k \times p}, \bm{\eta}_k, \bm{\Psi}_k, \nu_k)\).}
\nomenclature[N]{\(\tilde{\bm{\Sigma}}_k\)}{Shorthand for the scale matrix of the multivariate t-distribution modeling the predictive density \(p(\tilde{\bm{y}}_{\tilde{p}}|\tilde{z}_k = 1, \bm{Y}_{N_k \times p}, \bm{\eta}_k, \bm{\Psi}_k, \nu_k)\).}
\nomenclature[N]{\(F(\cdot)\)}{F-distribution cumulative distribution function.}
\nomenclature[N]{\(\tilde{\bm{\mu}}_k\)}{Predictive mean of \(\tilde{\bm{y}}_{\tilde{p}}\) in the multivariate t-distribution}
\nomenclature[N]{\(\tilde{\bm{\Sigma}}_k\)}{Predictive scale matrix of \(\tilde{\bm{y}}_{\tilde{p}}\) in the multivariate t-distribution}
\nomenclature[N]{\(\tilde{\nu}_k\)}{Degrees of freedom for predictive distribution \(p(\tilde{\bm{y}}_{\tilde{p}}|\tilde{z}_k = 1, \bm{Y}_{N_k\times p}, \bm{\eta}_k, \bm{\Psi}_k,\nu_k)\) related to event category \(k\)}
\nomenclature[N]{\(m\)}{Degrees of freedom for a Matrix t-distribution}
\nomenclature[N]{\(\mathcal{W}^{-1}(a, \bm{B})\)}{Inverse Wishart distribution with degrees of freedom \(a\) and scale matrix \(\bm{B}\)}
\nomenclature[N]{\(\mathcal{N}_{N,p}(\bm{M}, \bm{\Sigma} \otimes \bm{\Psi})\)}{Matrix normal distribution for a random \(N \times p\) matrix with mean \(\bm{M}\), row scale matrix \(\bm{\Sigma}\) and column scale matrix \(\bm{\Psi}\)}
%\nomenclature[N]{\(\mathrm{Categorical}(\bm{p}^{\top})\)}{Draw from a Categorical distribution with each category having probability defined by vector \(\bm{p}\)}
%\nomenclature[N]{\(\mathrm{Multinomial}(n, \bm{p}^{\top})\)}{An \(n\) draw from a Multinomial distribution with each category having probability defined by vector \(\bm{p}\)}
%\nomenclature[N]{\(\mathrm{Bernoulli}(\lambda)\)}{Bernoulli distribution with probability of success equal to \(\lambda\)}
%\nomenclature[N]{\(x \sim f(y)\)}{\(x\) behaves according to distribution \(f\) with parameters \(y\)}

%% Other O

\nomenclature[O]{\(\bm{1}_{N}\)}{Vector of ones of length \(N\)}
\nomenclature[O]{\(\bm{0}_{a \times b}\)}{Vector or matrix of zeros with \(a\) rows and \(b\) columns.}
\nomenclature[O]{\(\bm{J}_N\)}{Square matrix of ones with dimension \(N\)}
\nomenclature[O]{\(\bm{J}_{a \times b}\)}{Rectangular matrix of ones with \(a\) rows and \(b\) columns.}
\nomenclature[O]{\(\bm{I}_N\)}{Identity matrix of dimension \(N\)}

%% Dimension and Indexing I

\nomenclature[I]{\(k\)}{Index of event groups}
\nomenclature[I]{\(K\)}{Number of event groups}
\nomenclature[I]{\(N\)}{Total number of event observations in the training data}
\nomenclature[I]{\(N_k\)}{Number of event observations for the \(k^{\mathrm{th}}\) event group}
\nomenclature[I]{\(p\)}{Total number of discriminants considered for each event category.}
\nomenclature[I]{\(\tilde{p}\)}{Dimension of \(\tilde{\bm{y}}_{\tilde{p}}\).  \(\tilde{p}\) must be an integer which satisfies \(0 < \tilde{p} \leq p\).}
\nomenclature[I]{\(\ell\)}{Indexes the columns of \(\bm{Y}_{N_k \times p}\).}
\nomenclature[I]{\(N_k^m\)}{The number of missing entries in the \(\ell^{\mathrm{th}}\) column of \(\bm{Y}_{N_k \times p}\).  There are \(p\) values of \(N_k^m\) associated with each \(\bm{Y}_{N_k \times p}\), with the index often implied to simplify notation, but referred to as \((N_k^m)^{\ell}\) when necessary.}
\nomenclature[I]{\(N^{\mathrm{train}}\)}{Total number of training events which have zero missing values.  Used to describe the Monte Carlo experiments.}
\nomenclature[I]{\(N^{\mathrm{train}^{-}}\)}{Total number of training events which have one or more missing values.  Use to describe the Monte Carlo experiments.}
\nomenclature[I]{\(N^{\mathrm{test}}\)}{Total number of events used for testing to measure performance for the Monte Carlo experiments.  Includes event observations which have no missing data and some missing data.}
\nomenclature[I]{\(N_k^o\)}{The number of available entries in the \(\ell^{\mathrm{th}}\) column of \(\bm{Y}_{N_k \times p}\).  There are \(p\) values of \(N_k^o\) associated with each \(\bm{Y}_{N_k \times p}\), with the index often implied to simplify notation, but referred to as \((N_k^o)^{\ell}\) when necessary.}
\nomenclature[I]{\(\bm{N}\)}{A vector containing all values \([N_1, \dots , N_K]\).}

%%%%%%%%

\section{Nomenclature}\label{sec:nomencl}

\vspace{-1cm}

\printnomenclature

\section{The Statistical Model}\label{sec:the-model}

\subsection{Data Likelihood}\label{sec:normal-likelihood}

Let \((\bm{y}^{i}_k)_{p \times 1}\) be the \(i^{\mathrm{th}}\) of \(N_k\) observations of \(p\) discriminants from the \(k^{\mathrm{th}}\) event where \(k \in \{1, \dots , K\}\).  Each observation in the training data can be indexed as such.  There is a total of \(N_k\) observations for each event category, making the total number of event data \(N = \sum_{k = 1}^K N_k\).  We assume that each \((\bm{y}^{i}_k)_{p \times 1}\) is obtained from a multivariate normal distribution with a mean vector \(\bm{\mu}_k\) and covariance \(\bm{\Sigma}_k\).  Any random vector \(\bm{y}_p\) marginally drawn from the data is assumed to be sampled from a mixture of normal distributions,
\begin{equation}
\bm{y}_p \sim \sum_{k = 1}^{K} \pi_k \mathcal{N}(\bm{\mu}_k, \bm{\Sigma}_k)
\end{equation}
where \(\pi_k\) is a marginal probability of an event category and \(\sum_{k = 1}^K \pi_k = 1\).  Let \(\bm{Y}_{N \times p}\) denote all data observations, \(\bm{\mu} \equiv \{\bm{\mu}_1, \dots , \bm{\mu}_K\}\), \(\bm{\Sigma} \equiv \{\bm{\Sigma}_1, \dots , \bm{\Sigma}_K\}\), and \(\bm{\pi} \equiv \{\pi_1, \dots , \pi_K\}\).  Because the event category for each \(\bm{y}_k\) is already known, the likelihood is proportional to:
\begin{align}
p(\bm{Y}_{N \times p}|\bm{\mu}, \bm{\Sigma}, \bm{\pi}) &\propto \prod_{k = 1}^K \prod_{i = 1}^{N_k} \pi_k|\bm{\Sigma}_k|^{-\frac{1}{2}}\exp\left\{-\frac{1}{2}\left(\bm{y}^{i}_k - \bm{\mu}_k\right)^{\top} \bm{\Sigma}_k^{-1}\left(\bm{y}^{i}_k - \bm{\mu}_k\right)\right\}\notag \\
&\propto \prod_{k = 1}^K \pi_k^{N_k}|\bm{\Sigma}_k|^{-\frac{N_k}{2}}\exp\left\{-\frac{1}{2}\sum_{i = 1}^{N_k}\left(\bm{y}^{i}_k - \bm{\mu}_k\right)^{\top} \bm{\Sigma}_k^{-1}\left(\bm{y}^{i}_k - \bm{\mu}_k\right)\right\}
\label{eq:likelihood}
\end{align}

\subsection{Prior Distributions}\label{sec:priors}

We choose conjugate prior \citep{hoff2009first} distributions for their mathematical properties, advantageous for later derivations.  A prior distribution is required for the parameters \(\pi_k, \bm{\Sigma}_k\) and \(\bm{\mu}_k\) associated with each event category.  We specify the conditional priors \\\(p(\bm{\mu}_k, \bm{\Sigma}_k) = p(\bm{\mu}_k|\bm{\Sigma}_k, \bm{\eta}_k)p(\bm{\Sigma}_k|\bm{\Psi}_k, \nu_k)\).  \(p(\bm{\mu}_k|\bm{\Sigma}_k, \bm{\eta}_k)\) is a multivariate normal distribution with a mean vector of \(\bm{\eta}_k\) and a covariance of \(\bm{\Sigma}_k\), while \(p(\bm{\Sigma}_k|\bm{\Psi}_k, \nu_k)\) is an inverse Wishart distribution \citep{hoff2009first, gupta2018matrix} with a scale matrix of \(\bm{\Psi}_k\) and degrees of freedom \(\nu_k\).  The prior on the vector of mixture parameters \(p(\bm{\pi}|\bm{\alpha})\) is a Dirichlet distribution \citep{ng2011dirichlet} with a parameter vector \(\bm{\alpha}\) of length \(K\).

The choice of these prior distributions allows for closed form posterior distributions of the parameters \(\bm{\pi}, \bm{\mu}_k\), and \(\bm{\Sigma}_k\), as well as closed form integrals for the predictive distributions \(p(\tilde{\bm{y}}_{\tilde{p}}|\bm{Y}_{N_k\times p}, \bm{\Psi}_k, \bm{\eta}_{k}, \nu_k, \tilde{z}_k = 1)\) and \(p(\tilde{z}_k = 1|\bm{Y}_{N\times p}, \bm{\alpha})\).

In our experiments, we use the default prior parameter settings from the \textsf{R} package \texttt{ezECM} for a logit transformation of the data.  The default for \(p(\bm{\mu}_k|\bm{\Sigma}_k, \bm{\eta}_k)\) is \(\bm{\eta}_k = \bm{0}_{p\times 1}\).  In combination with the values for \(p(\bm{\Sigma}_k|\bm{\Psi}_k, \nu_k)\) of \(\bm{\Psi}_k = \bm{I}_p\) and \(\nu_k = p\), the marginal likelihood of the data becomes fairly diffuse and is in fact matrix variate Cauchy \citep{gupta2018matrix}, where the expectation and the variance are undefined. For prediction, the degrees of freedom for category \(k\) then becomes \(N_k + 1\).  The the setting of \(\bm{\alpha}\) in the prior distribution \(p(\bm{\pi}|\bm{\alpha})\) is \(\bm{\alpha} = \bm{1}_K \times 1/2\), equivalent to a fairly non-informative Jeffreys and reference prior \citep{yang1996catalog}.

\subsection{Marginal Likelihood of the Event Data}\label{sec:data-likelihood}

To reduce computation time we choose to utilize analytical solutions to integrals instead of numerical solutions when reasonable.  Deriving the marginal likelihood of the \(k^{\mathrm{th}}\) event data \(p(\bm{Y}_{N_k \times p}|\bm{\eta}_k, \bm{\Psi}_k, \nu_k)\) is useful for later deriving the predictive distribution \(p(\tilde{\bm{y}}|\tilde{z}_k = 1, \bm{Y}_{N_k\times p}, \bm{\eta}_k, \bm{\Psi}_k,\nu_k)\) in Appendix \ref{sec:predictive-y}.

Each \(p(\bm{Y}_{N_k \times p}|\bm{\eta}_k, \bm{\Psi}_k, \nu_k)\) requires integration over \(\bm{\mu}_k \in \mathbb{R}^p\) and \(\bm{\Sigma}_k\) over a space of symmetric positive definite matrices such that:
\begin{equation}
p(\bm{Y}_{N_k \times p}|\bm{\eta}_k, \bm{\Psi}_k, \nu_k) = \int_{\bm{\Sigma}} \int_{\mathbb{R}^p} p(\bm{Y}_{N_k \times p}|\bm{\mu}_k, \bm{\Sigma}_k) p(\bm{\mu}_k, \bm{\Sigma}_k| \bm{\eta}_k, \bm{\Psi}_k, \nu_k)d\bm{\mu}_k d\bm{\Sigma}_k.
\end{equation}

The stepwise solution to this integral requires the trace identities \(tr(\bm{A}\bm{B}\bm{C}) = tr(\bm{C}\bm{A}\bm{B})\) for matrices conformal for multiplication \citep{harville1998matrix}, as well as the matrix determinant lemma \citep{harville1998matrix}.
\begin{equation}
\begin{split}
&p(\bm{Y}_{N_k \times p}|\bm{\eta}_k, \bm{\Psi}_k, \nu_k)\\
&=  \int_{\bm{\Sigma}_k} \int_{\bm{\mu}_k} p(\bm{Y}_{N_k \times p}|\bm{\mu}_k, \bm{\Sigma}_k) p(\bm{\mu}_k|\bm{\Sigma}_k, \bm{\eta}_k)p(\bm{\Sigma}_k|\bm{\Psi}_k, \nu_k)d\bm{\mu}_k d\bm{\Sigma}_k\\
&\propto \int_{\bm{\Sigma}_k} \int_{\bm{\mu}_k}\left[\vphantom{|\bm{\Sigma}_k|^{-\frac{(\nu_k + p + 1)}{2}}}|\bm{\Sigma}_k|^{-\frac{N_k}{2}}\exp\left\{-\frac{1}{2}\sum_{i = 1}^{N_k}\left(\bm{y}^{i}_k - \bm{\mu}_k\right)^{\top} \bm{\Sigma}_k^{-1}\left(\bm{y}^{i}_k - \bm{\mu}_k\right)\right\}|\bm{\Sigma}_k|^{-\frac{1}{2}}\right.\\
&\qquad\left.\vphantom{\exp\left\{-\frac{1}{2}\sum_{i = 1}^{N_k}\left(\bm{y}^{i}_k - \bm{\mu}_k\right)^{\top}\right\}}\exp\left\{-\frac{1}{2}\left(\bm{\mu}_k - \bm{\eta}_k\right)^{\top}\bm{\Sigma}_k^{-1}\left(\bm{\mu}_k - \bm{\eta}_k\right)\right\}|\bm{\Sigma}_k|^{-\frac{(\nu_k + p + 1)}{2}}e^{-\frac{1}{2}\mathrm{tr}\left(\bm{\Psi}_k \bm{\Sigma}_k^{-1}\right)}\right]d\bm{\mu}_k d\bm{\Sigma}_k\\
&\propto \int_{\bm{\Sigma}_k}|\bm{\Sigma}_k|^{-\frac{(N_k + \nu_k + p + 1)}{2}} \\
& \qquad \exp\left\{-\frac{1}{2}\mathrm{tr}\left(\bm{\Sigma}_k^{-1} \left(\begin{bmatrix}
\bm{Y}_{p \times N_k}^{\top} & \bm{\eta}_k
\end{bmatrix}
\left(\bm{I}_{N_k+1} - \frac{1}{N_k + 1} \bm{J}_{N_k+1}\right)
\begin{bmatrix}
\bm{Y}_{N_k\times p} \\
\bm{\eta}_k^{\top}
\end{bmatrix}
 + \bm{\Psi}_k
\right)\right)\right\} d\bm{\Sigma}_k\\
&\propto  \left|\left(\bm{Y}_{N_k \times p} - \bm{1}_{N_k}\bm{\eta}_k^{\top}\right)^{\top}\left(\bm{I}_{N_k} - \frac{1}{N_k + 1} \bm{J}_{N_k} \right) \left(\bm{Y}_{N_k \times p} - \bm{1}_{N_k}\bm{\eta}_k^{\top}\right) + \bm{\Psi}_k\right|^{-\frac{(N_k + \nu_k)}{2}}\\
&\propto \left|\bm{I}_{N_k} + \left(\bm{I}_{N_k} - \frac{1}{N_k + 1} \bm{J}_{N_k} \right)\left(\bm{Y}_{N_k \times p} - \bm{1}_{N_k}\bm{\eta}_k^{\top}\right) \bm{\Psi}_k^{-1} \left(\bm{Y}_{N_k \times p} - \bm{1}_{N_k}\bm{\eta}_k^{\top}\right)^{\top} \right|^{-\frac{((\nu_k + 1 - p) + N_k + p - 1)}{2}}
\end{split}
\end{equation}
We can recognize the above as proportional to a matrix t-distribution, where \(\bm{Y}_{N_k \times p} \sim T_{N_k, p}(\nu_k + 1 - p, \bm{1}_{N_k}(\bm{\eta}_k)^{\top}, \bm{I}_{N_k} + \bm{J}_{N_k}, \bm{\Psi}_k)\).  The choice of \(\nu_k\) has an effect on the properties of this distribution, and later derivations.  If the degrees of freedom, \(\nu_k + 1 -p\), are equal to one the distribution is matrix variate Cauchy \citep{gupta2018matrix}.  The degrees of freedom must be greater than 2 for the variance to be defined \citep{gupta2018matrix}.

\subsection{Conditional Missing Data Distributions}\label{sec:missing-data-conditionals}

There is typically not a simple closed form for \(p(\bm{Y}^{-}_k|\bm{Y}^{+}_k, \bm{\eta}_k, \bm{\Psi}_k, \nu_k)\), although we will cover the case in this section when a simple density is available.  Details on the hyperparameters \(\bm{\eta}, \bm{\Psi}\) and \(\bm{\nu}\) can be found in Appendices \ref{sec:priors} and \ref{sec:data-likelihood}.  Instead, we derive distributions of subsets of the missing data, conditional on all other data, which are available in closed form as recognizable densities.  We choose each subset of \(\bm{Y}^{-}_k\) to be the full set of missing entries found in a single column of \(\bm{Y}_{N_k \times p}\).  Depending on the application details using a single row may be advantageous, and in unlikely cases the set can be matrix of missing entries instead of a vector.  The end goal is implementation within a Gibbs sampler \citep{hoff2009first, robert1999monte} to obtain marginal samples from \(p(\bm{Y}^{-}_k|\bm{Y}^{+}_k, \bm{\eta}_k, \bm{\Psi}_k, \nu_k)\).  With the goal of numerical integration in mind, the derivation for the posterior distributions required for MCMC are as follows.

First, partition \(\bm{Y}_{N_k \times p}, \bm{\eta}_k\), and \(\bm{\Psi}_k\) in a similar manner to Appendix \ref{sec:matrix-t}
\begin{equation}
\bm{Y}_{N_k \times p} = \raisebox{-0.6666666\baselineskip}{$\begin{blockarray}{ccc}
\begin{block}{[cc]c}
\bm{y}^-_{N_k^m \times 1} & \bm{Y}_{N^m_k \times (p - 1)}\bigstrut[t] & N^m_k \\
\bm{y}^{+}_{N_k^o \times 1} & \bm{Y}_{N^o_k \times (p - 1)} \bigstrut[b] & N^o_k \\
\end{block}
1 & (p-1) & 
\end{blockarray}$} , \  \bm{\eta}_k = \raisebox{-0.3333333\baselineskip}{$\begin{blockarray}{cc}
\begin{block}{[c]c}
\eta_k^{-} \bigstrut[t] & 1 \\
\bm{\eta}_k^{+} \bigstrut[b] & (p-1)\\
\end{block}
\end{blockarray}$}, \bm{\Psi}_k = \raisebox{-0.6666666\baselineskip}{$\begin{blockarray}{ccc}
\begin{block}{[cc]c}
\Psi_{k(11)} & \bm{\Psi}_{k(12)} \bigstrut[t] & 1\\
\bm{\Psi}_{k(21)} & \bm{\Psi}_{k(22)} \bigstrut[b] & (p-1)\\
\end{block}
1 & (p-1) & 
\end{blockarray}$}
\end{equation}
Using this new notation, we want to derive the conditional distribution \\ \(p(\bm{y}^{-}_{N_k^m \times 1}|\bm{y}^{+}_{N_k^o \times 1}, \bm{Y}_{N_k^m \times (p-1)}, \bm{Y}_{N_k^o \times (p-1)}, \bm{\eta}_k, \bm{\Psi}_k, \nu_k)\).  \(\bm{y}^{-}_{N_k^m \times 1}\) are the missing entries in a single column of \(\bm{Y}_{N_k \times p}\).  The number of missing entries in this column are \(N_k^m\) and the number of observed entries in the same column are \(N_k^o\) such that \(N_k^m + N_k^o = N_k\).  The remaining \(p-1\) columns are represented by the matrices \(\bm{Y}_{N_k^m \times (p-1)}\) and \(\bm{Y}_{N_k^o \times (p-1)}\).  While some of the entries in these matrices are in fact missing, we will condition on all entries represented by these matrices, where the values are imputed in other steps of the MCMC.  See Algorithm \ref{alg:training}.

To obtain  \(p(\bm{y}^{-}_{N_k^m \times 1}|\bm{y}^{+}_{N_k^o \times 1}, \bm{Y}_{N_k^m \times (p-1)}, \bm{Y}_{N_k^o \times (p-1)}, \bm{\eta}_k, \bm{\Psi}_k, \nu_k)\), first we use the conditional equations in Appendix \ref{sec:matrix-t} to obtain \(p(\bm{y}^{-}_{N_k^m \times 1}, \bm{y}^{+}_{N_k^o \times 1}| \bm{Y}_{N_k^m \times (p-1)}, \bm{Y}_{N_k^o \times (p-1)}, \bm{\eta}_k, \bm{\Psi}_k, \nu_k)\).
\begin{equation}
\begin{split}
& p(\bm{y}^{-}_{N_k^m \times 1}, \bm{y}^{+}_{N_k^o \times 1}| \bm{Y}_{N_k^m \times (p-1)}, \bm{Y}_{N_k^o \times (p-1)}, \bm{\eta}_k, \bm{\Psi}_k, \nu_k) =\\
& \qquad T_{N_k, 1}\left(\nu_k,  \eta_k^{-}\bm{1}_{N_k} + \left(\begin{bmatrix}
\bm{Y}_{N^m_k \times (p - 1)}\\
 \bm{Y}_{N^o_k \times (p - 1)} 
\end{bmatrix} - \bm{1}_{N_k} (\bm{\eta}_k^{+})^{\top}\right)\bm{\Psi}_{k(22)}^{-1}\bm{\Psi}_{k(21)}, \right.\\
&\qquad \qquad \left. \bm{I}_{N_k} + \bm{J}_{N_k} + \left(\begin{bmatrix}
\bm{Y}_{N^m_k \times (p - 1)}\\
 \bm{Y}_{N^o_k \times (p - 1)} 
\end{bmatrix} - \bm{1}_{N_k} (\bm{\eta}_k^{+})^{\top}\right)\bm{\Psi}_{k(22)}^{-1} \left(\begin{bmatrix}
\bm{Y}_{N^m_k \times (p - 1)}\\
 \bm{Y}_{N^o_k \times (p - 1)} 
\end{bmatrix} - \bm{1}_{N_k} (\bm{\eta}_k^{+})^{\top}\right)^{\top}, \right.\\
&\qquad \qquad \left. \vphantom{\left(\begin{bmatrix}
\bm{Y}_{N^m_k \times (p - 1)}\\
\bm{Y}_{N^o_k \times (p - 1)} 
\end{bmatrix}\right)}\Psi_{k(11)} - \bm{\Psi}_{k(12)}\bm{\Psi}_{k(22)}^{-1}\bm{\Psi}_{k(21)}\right)
\end{split}
\end{equation}
In the event that \(N_k^m = N_k\) then no further conditioning is required.  Otherwise, the conditional t-distribution equations in Appendix \ref{sec:matrix-t} are used again to find \\
\(p(\bm{y}^{-}_{N_k^m \times 1}|\bm{y}^{+}_{N_k^o \times 1}, \bm{Y}_{N_k^m \times (p-1)}, \bm{Y}_{N_k^o \times (p-1)}, \bm{\eta}_k, \bm{\Psi}_k, \nu_k)\), given by 
\begin{align}
\bm{y}^-_{N_k^m \times 1}| \bm{y}^{+}_{N_k^o \times 1}, \bm{Y}_{N^m_k \times (p - 1)},   \bm{Y}_{N^o_k \times (p - 1)} &\sim T_{N_{k}^m, 1}\left(\nu_k + N_k^o,\bm{M}^{1r} + \bm{\Lambda}_{12}\bm{\Lambda}_{22}^{-1}(\bm{y}^{+}_{N_k^o \times 1}- \bm{M}^{2r}),\right. \notag\\
&\qquad \qquad \bm{\Lambda}_{11} - \bm{\Lambda}_{12}\bm{\Lambda}_{22}^{-1}\bm{\Lambda}_{21}, \notag\\
&\qquad \qquad \left. \Omega + (\bm{y}^{+}_{N_k^o \times 1} - \bm{M}^{2r})^{\top}\bm{\Lambda}_{22}^{-1}(\bm{y}^{+}_{N_k^o \times 1} - \bm{M}^{2r})\right),
\end{align}
where
\begin{align}
\bm{M}^{1r} &= \eta_k^{-} \bm{1}_{N_k^m} + ( \bm{Y}_{N^m_k \times (p - 1)} - \bm{1}_{N_k^m}(\bm{\eta}_k^{+})^{\top})\bm{\Psi}_{k(22)}^{-1}\bm{\Psi}_{k(21)}\\
\bm{M}^{2r} &= \eta_k^{-} \bm{1}_{N_k^o} + (\bm{Y}_{N^o_k \times (p - 1)} - \bm{1}_{N_k^o}(\bm{\eta}_k^{+})^{\top})\bm{\Psi}_{k(22)}^{-1}\bm{\Psi}_{k(21)}\\
\bm{\Lambda}_{11} &= ( \bm{Y}_{N^m_k \times (p - 1)} - \bm{1}_{N_k^m}(\bm{\eta}_k^{+})^{\top})\bm{\Psi}_{k(22)}^{-1}( \bm{Y}_{N^m_k \times (p - 1)} - \bm{1}_{N_k^m}(\bm{\eta}_k^{+})^{\top})^{\top} + \bm{I}_{N_k^{m}} + \bm{J}_{N_k^{m}}\\
\bm{\Lambda}_{12} &= ( \bm{Y}_{N^m_k \times (p - 1)} - \bm{1}_{N_k^m}(\bm{\eta}_k^{+})^{\top})\bm{\Psi}_{k(22)}^{-1} (\bm{Y}_{N^o_k \times (p - 1)} - \bm{1}_{N_k^o}(\bm{\eta}_k^{+})^{\top})^{\top}+ \bm{J}_{N_k^m \times N_k^{o}}\\
\bm{\Lambda}_{21} & = \bm{\Lambda}_{12}^{\top}\\
\bm{\Lambda}_{22} &= (\bm{Y}_{N^o_k \times (p - 1)} - \bm{1}_{N_k^o}(\bm{\eta}_k^{+})^{\top})\bm{\Psi}_{k(22)}^{-1}(\bm{Y}_{N^o_k \times (p - 1)} - \bm{1}_{N_k^o}(\bm{\eta}_k^{+})^{\top})^{\top} + \bm{I}_{N_k^{o}} + \bm{J}_{N_k^o}\\
\Omega &= \Psi_{k(11)} - \bm{\Psi}_{k(12)}\bm{\Psi}_{k(22)}^{-1}\bm{\Psi}_{k(21)}.
\end{align}

\subsection{Predictive Distributions}

These Appendix sections detail the derivations for the densities on the right hand side of
\begin{equation}
p(\tilde{z}_k = 1|\tilde{\bm{y}}_{\tilde{p}}, \bm{Y}_{N \times p}, \bm{\eta}, \bm{\Psi}, \bm{\nu}, \bm{\alpha}) = \frac{p(\tilde{\bm{y}}_{\tilde{p}}|\tilde{z}_k = 1, \bm{Y}_{N_k\times p}, \bm{\eta}_k, \bm{\Psi}_k,\nu_k)p(\tilde{z}_k = 1|\bm{Y}_{N \times p},\bm{\alpha})}{p(\tilde{\bm{y}}_{\tilde{p}}|\bm{Y}_{N\times p}, \bm{\eta}, \bm{\Psi}, \bm{\nu}, \bm{\alpha})},
\end{equation}
where \(p(\tilde{\bm{y}}_{\tilde{p}}|\tilde{z}_k = 1, \bm{Y}_{N_k\times p}, \bm{\eta}_k, \bm{\Psi}_k,\nu_k)\) is the predictive distribution for \(\tilde{\bm{y}}_{\tilde{p}}\), \(p(\tilde{z}_k = 1|\bm{Y}_{N \times p},\bm{\alpha})\) is the predictive distribution of category \(k\), and \(p(\tilde{\bm{y}}_{\tilde{p}}|\bm{Y}_{N\times p}, \bm{\eta}, \bm{\Psi}, \bm{\nu}, \bm{\alpha})\) is the predictive distribution marginalized over all \(K\) categories.

\subsubsection{\(p(\tilde{\bm{y}}_{\tilde{p}}|\tilde{z}_k = 1, \bm{Y}_{N_k\times p}, \bm{\eta}_k, \bm{\Psi}_k,\nu_k)\)}\label{sec:predictive-y}

The predictive distribution of \(\tilde{\bm{y}}_{\tilde{p}}\) can be found by solving for the integral
\begin{align}
p(\tilde{\bm{y}}_{\tilde{p}}|\tilde{z}_k = 1, \bm{Y}_{N_k \times p},\bm{\eta}_k, \bm{\Psi}_k, \nu_k) =& \int_{\bm{\Sigma}} \int_{\mathbb{R}^p} p(\tilde{\bm{y}}_{\tilde{p}}|\tilde{z}_k = 1, \bm{\mu}_k, \bm{\Sigma}_k) \notag \\
& \qquad p(\bm{\mu}_k, \bm{\Sigma}_k| \bm{Y}_{N_k \times p}, \bm{\eta}_k, \bm{\Psi}_k, \nu_k)d\bm{\mu}_k d\bm{\Sigma}_k,
\end{align}
and then by using the properties of the marginal multivariate t-distribution for \(\tilde{p} < p\).  However, because we already have \(p(\bm{Y}_{N_k \times p}|\tilde{z}_k = 1, \bm{\eta}_k, \bm{\Psi}_k, \nu_k)\), and the properties of conditional distributions for matrix t-distributed random variables, we choose to exploit conditional probability for this illustration; \\ \(p(\tilde{\bm{y}}_{p}, \bm{Y}_{N_k \times p}|\tilde{z}_k = 1, \bm{\eta}_k, \bm{\Psi}_k, \nu_k) = p(\tilde{\bm{y}}_{p}| \tilde{z}_k = 1, \bm{Y}_{N_k \times p}, \bm{\eta}_k, \bm{\Psi}_k, \nu_k)p(\bm{Y}_{N_k \times p}|\tilde{z}_k = 1, \bm{\eta}_k, \bm{\Psi}_k, \nu_k)\).  Note that we switch to the notation \(\tilde{\bm{y}}_{p}\) from \(\tilde{\bm{y}}_{\tilde{p}}\) to indicate a full \(\tilde{p} = p\) dimensional vector.  Properties of the marginal distribution of a multivariate t-distribution are applied at the end for the the case where \(\tilde{p} < p\).

Given the result in Appendix \ref{sec:data-likelihood}  , we can write the joint distribution \(p(\tilde{\bm{y}}_{p}, \bm{Y}_{N_k \times p}|\tilde{z}_k = 1, \bm{\eta}_k, \bm{\Psi}_k, \nu_k)\) as 
\begin{align}
&p(\tilde{\bm{y}}_{p}, \bm{Y}_{N_k \times p}|\tilde{z}_k = 1, \bm{\eta}_k, \bm{\Psi}_k, \nu_k) \propto \notag\\
&\qquad\left|\bm{I}_{N_k + 1} + \left(\bm{I}_{N_k + 1} - \frac{1}{N_k + 1 + 1} \bm{J}_{N_k + 1} \right)\right. \notag \\
&\left.\left(\begin{bmatrix}
\tilde{\bm{y}}_{p}^{\top}\\
\bm{Y}_{N_k \times p}
\end{bmatrix} - \bm{1}_{N_k + 1}\bm{\eta}_k^{\top}\right) \bm{\Psi}_k^{-1} \left(\begin{bmatrix}
\tilde{\bm{y}}_{p}^{\top}\\
\bm{Y}_{N_k \times p}
\end{bmatrix} - \bm{1}_{N_k + 1}\bm{\eta}_k^{\top}\right)^{\top} \right|^{-\frac{((\nu_k + 1 - p) + N_k + 1 + p - 1)}{2}},
\end{align}
which is \(T_{(N_k + 1) \times p}(\nu_k + 1 - p, \bm{1}_{N_k +1}\bm{\eta}^{\top}_k, \bm{I}_{N_k + 1} + \bm{J}_{N_k + 1}, \bm{\Psi}_k)\).  Then the conditioning equations in Appendix \ref{sec:matrix-t} can be used to provide \(p(\tilde{\bm{y}}_{p}|\tilde{z}_k = 1, \bm{Y}_{N_k\times p}, \bm{\eta}_k, \bm{\Psi}_k,\nu_k)\).  Rearranging terms as well as using the relationship between the matrix t-distribution and the multivariate t-distribution in Appendix \ref{sec:matrix-t} produces a density proportional to a multivariate t-distribution noted as \(t_{p}(N_k + \nu_k + 1 -p, \tilde{\bm{\mu}}_k, \tilde{\bm{\Sigma}}_k)\), where \(\tilde{\bm{\mu}}_k = (N_k\bar{\bm{y}}_k + \bm{\eta}_k) / (N_k + 1)\) and
{\footnotesize
\begin{equation}
\tilde{\bm{\Sigma}}_k = \frac{N_k + 2}{(N_k + 1)(N_k + \nu_k + 1 - p )}\left(\bm{\Psi}_k + (\bm{Y}_{N_k \times p} - 1_{N_k}\bm{\eta}_k^{\top})^{\top}\left(\bm{I}_{N_k} - \frac{1}{N_k + 1}\bm{J}_{N_k}\right)(\bm{Y}_{N_k \times p} - 1_{N_k}\bm{\eta}_k^{\top})\right).
\end{equation}
}
Additionally, \(\bar{\bm{y}}_k = \bm{Y}_{N_k \times p}^{\top} \bm{1}_{N_k} / N_k \). 

For the case when a new observation \(\tilde{\bm{y}}_{\tilde{p}}\) occurs, and \(\tilde{p} < p\) the properties of a multivariate t-distribution \citep{kotz2004multivariate} are used to find the marginal distribution of this shorter vector.  In short, the marginal distribution of \(\tilde{\bm{y}}_{\tilde{p} < p}\) has the same degrees of freedom as the distribution of \(\tilde{\bm{y}}_{\tilde{p} = p}\), but the elements of the mean vector and scale matrix corresponding to the missing elements of \(\tilde{\bm{y}}_{p}\) are removed.

\subsubsection{\(p(\tilde{z}_k = 1|\bm{Y}_{N \times p}, \bm{\alpha})\)}\label{sec:predictive-a-priori-y}

This distribution can be thought of as the probability that one of the categories will occur, before \(\tilde{\bm{y}}_{\tilde{p}}\) is observed.  In simple terms, it is the fraction of observations in the \(k^{\mathrm{th}}\) category in the training data, modified by the prior parameter \(\bm{\alpha}\).  In the event that the training data is not an accurate reflection of the frequency of future events, choosing \(p(\tilde{z}_k = 1|\bm{\alpha})\), uninformed by the training data may be preferable.  An option in such a scenario is \(p(\tilde{z}_k = 1| \alpha_1 = \dots = \alpha_K) = 1/K \), but any subjective choice under the constraint \(\sum_{i = 1}^K p(\tilde{z}_k =1 |\bm{\alpha}) = 1\) is valid.  Here we focus on the result when conditioning on the training data.%Here we derive \(p(\tilde{\bm{z}}_K|\bm{Y}_{N \times p}, \bm{\alpha})\), and 

The vector of weights \(\bm{\pi}\) is present in the likelihood shown in Appendix \ref{sec:normal-likelihood}. However, from Equation (\ref{eq:likelihood}), the joint posterior of the parameters can be written as
\begin{equation}
p(\bm{\mu}, \bm{\Sigma}, \bm{\pi}|\bm{Y}_{N \times p}, \bm{\eta}, \bm{\Psi}, \bm{\nu}, \bm{\alpha}) = p(\bm{\pi}|\bm{Y}_{N \times p}, \bm{\alpha})p(\bm{\mu}, \bm{\Sigma}|\bm{Y}_{N \times p}, \bm{\eta}, \bm{\Psi}, \bm{\nu}),
\end{equation}
meaning \(\bm{\pi}\) is conditionally independent of the other parameters.  Subsequently, we note \(p(\bm{\pi}|\bm{N}) = p(\bm{\pi}|\bm{Y}_{N \times p})\), where \(\bm{N} = [N_1 , \dots , N_K]\), because the number of training observations in each event category is the only data that informs the posterior \(p(\bm{\pi}|\bm{N}, \bm{\alpha})\), and therefore \(p(\tilde{\bm{z}}_K| \bm{N}, \bm{\alpha})\).  The posterior \(p(\bm{\pi}|\bm{N}, \bm{\alpha})\) can be recognized as the Dirichlet distribution \citep{ng2011dirichlet} \(\mathrm{Dir}(N_1 + \alpha_1, \dots, N_K + \alpha_K)\).  A draw of the random variable \(\bm{\pi}\) is the probability of each category, based on counts of each category from the likelihood and the choice of the vector \(\bm{\alpha}\).

The variable encoding the \emph{predictive category} is \(\tilde{\bm{z}}_K\).  The probability distribution \(p(\tilde{\bm{z}}_K|\bm{\pi})\) is a categorical distribution with category probabilities of \(\bm{\pi}\).  With the form of \(p(\tilde{\bm{z}}_K|\bm{\pi})\) and \(p(\bm{\pi}|\bm{N}, \bm{\alpha})\) specified, \(p(\tilde{\bm{z}}_K|\bm{Y}_{N\times p}, \bm{\alpha})\) is the result of integration
 \begin{equation}
 p(\tilde{\bm{z}}_K| \bm{N}, \bm{\alpha}) = \int_{[0,1]^K} p(\tilde{\bm{z}}_K| \bm{N}, \bm{\pi}, \bm{\alpha})p(\bm{\pi}|\bm{N}, \bm{\alpha})d\bm{\pi}.
 \end{equation}
The analytical form of the vector of probabilities \(p(\tilde{\bm{z}}_K|\bm{N}, \bm{\alpha})\), and thereby a specific \(p(\tilde{z}_k = 1|\bm{N}, \bm{\alpha})\), is found as
\begin{align}
 p(\tilde{\bm{z}}_K| \bm{N}, \bm{\alpha}) &= \int_{[0,1]^K} \prod_{k =1}^K \left[\pi^{\tilde{z}_k}_k\right] \Gamma\left(\sum_{k = 1}^K (N_k + \alpha_k)\right)\prod_{k =1}^K\frac{\pi_k^{N_k + \alpha_k - 1}}{\Gamma(N_k + \alpha_k)}d\bm{\pi} \notag \\
 &= \frac{\Gamma\left(\sum_{k = 1}^K (N_k + \alpha_k)\right)}{\prod_{k = 1}^K \Gamma(N_k + \alpha_k)}\int_{[0,1]^K} \prod_{k =1}^K \pi^{N_k + \alpha_k + \tilde{z}_k - 1}_k d\bm{\pi} \notag \\
 &= \frac{\Gamma\left(\sum_{k = 1}^K (N_k + \alpha_k)\right)}{\prod_{k = 1}^K \Gamma(N_k + \alpha_k)}\frac{\prod_{k = 1}^K \Gamma(N_k + \alpha_k + \tilde{z}_k)}{\Gamma\left(\sum_{k = 1}^K (N_k + \alpha_k + \tilde{z}_k)\right)} \notag \\
 &\tilde{z}_k = 1 \implies \notag\\
 p(\tilde{z}_k &= 1|\bm{N}, \bm{\alpha}) = \frac{N_k + \alpha_k}{\sum_{k = 1}^K (N_k + \alpha_k)},
\end{align}
where \(\Gamma(\cdot)\) is the gamma function which has the property \(\Gamma(x + 1) = x\Gamma(x)\) utilized in the simplification.  This result is a simplification of the Dirichlet-multinomial distribution \citep{ng2011dirichlet} for a single predictive draw, where a single \(\tilde{z}_k = 1\) for an arbitrary \(k\).

\subsubsection{\(p(\tilde{\bm{y}}_{\tilde{p}}|\bm{Y}_{N\times p}, \bm{\eta}, \bm{\Psi}, \bm{\nu}, \bm{\alpha})\)}\label{sec:marginal-predictive}

This density is a marginalization over all \(K\) categories to produce the total predictive density of observing \(\tilde{\bm{y}}_{\tilde{p}}\) from those categories.  Using the law of total probability, the result is simply
\begin{equation}
p(\tilde{\bm{y}}_{\tilde{p}}|\bm{Y}_{N \times p}, \bm{\eta}, \bm{\Psi}, \bm{\nu}, \bm{\alpha}) = \sum_{k = 1}^K p(\tilde{\bm{y}}_{\tilde{p}}|\tilde{z}_k = 1, \bm{Y}_{N_k \times p}, \bm{\eta}_k, \bm{\Psi}_k, \nu_k)p(\tilde{z}_k = 1| \bm{Y}_{N \times p}, \bm{\alpha}).
\end{equation}

%%%%%%%%%%%%%%%%%%%%%%%%%%%%%%%%%%%%%%%%

\section{Multivariate and Matrix t-Distributions}\label{sec:matrix-t}

The properties of the multivariate t-distribution are explored in great detail in \citet{kotz2004multivariate} with the relation to the matrix variate t-distribution summarized in \citet{gupta2018matrix}.  Herein the main parameterization used in \citet{kotz2004multivariate} is set as the standard, where a random vector \(\bm{t}_{p \times 1}\) is said to be multivariate t-distributed, \(t_{\nu}(\bm{\mu}, \bm{\Sigma})\), if the probability density function is given as
\begin{equation}
\frac{\Gamma((\nu + p)/2)}{(\pi \nu)^{\frac{p}{2}}\Gamma(\nu/2)}|\bm{\Sigma}|^{-\frac{1}{2}}\left[1 + \frac{1}{\nu}(\bm{t} - \bm{\mu})^{\top}\bm{\Sigma}^{-1}(\bm{t} - \bm{\mu})\right]^{-\frac{(\nu + p)}{2}},
\end{equation}
where \(\bm{\mu}_{p \times 1}\) is a \(p\) dimensional scale vector, \(\nu\) is the degrees of freedom parameter, \(\bm{\Sigma}\) is a scale or correlation matrix, and \(\Gamma(x)\) is the gamma function of \(x\).  When \(\nu > 1, \ \mathbb{E}[\bm{t}] = \bm{\mu}\), and when \(\nu > 2, \mathbb{C}\mathrm{ov}(\bm{t}) = \nu \bm{\Sigma} / (\nu - 2)\), otherwise these moments are undefined.

We utilize the form of the matrix variate t-distribution density in \citet{gupta2018matrix}, with some slight changes in notation from this reference.  The random variate \(\bm{T}_{N \times p} \) is matrix variate t-distributed \(T_{N \times p}(m, \bm{M}, \bm{\Sigma}, \bm{\Omega}) \) with a probability density function
\begin{equation}
\frac{\Gamma_p[\frac{1}{2}(N + m + p -1)]}{\pi^{\frac{1}{2}Np}\Gamma_p[\frac{1}{2}(m + N -1)]}|\bm{\Sigma}|^{-\frac{p}{2}}|\bm{\Omega}|^{-\frac{N}{2}}|\bm{I}_N + \bm{\Sigma}^{-1}(\bm{T} - \bm{M})\bm{\Omega}^{-1}(\bm{T}- \bm{M})^{\top}|^{-\frac{(N + m + p -1)}{2}},
\end{equation}
where \(m\) is the degrees of freedom, \(\bm{M}\) is a \(N \times p\) location matrix, and \(\bm{\Omega}_{p \times p}, \bm{\Sigma}_{N \times N}\) are positive definite symmetric scale matrices.  Analogous to the multivariate t-distribution, \(\mathbb{V}\mathrm{ar}(\mathrm{vec}(\bm{T})) = \bm{\Sigma}\otimes \bm{\Omega}/(m - 2)\), but is undefined when \(m \leq 2\), and \(\mathrm{vec}(\bm{T})\) is the vectorization operator.  Additionally, note that \(T_{1 \times p}(\nu, \bm{\mu}^{\top}, \sigma, \bm{\Omega}) = t_\nu(\bm{\mu}, \sigma \bm{\Omega} / \nu)\).

The following theorem is taken from \citet{gupta2018matrix} and utilized in Algorithm \ref{alg:training}.  Let \(\bm{A}_{N \times N}\) and \(\bm{B}_{p \times p}\) be non-singular square matrices, and let \(\bm{T}_{N\times p} \sim T_{N \times p}(m, \bm{M}, \bm{\Sigma}, \bm{\Omega})\).  Then, \(\bm{A}_{N \times N}\bm{T}_{N\times p}\bm{B}_{p \times p} \sim T_{N \times p}(m, \bm{A}\bm{M}\bm{B}, \bm{A}\bm{\Sigma}\bm{A}^{\top}, \bm{B}\bm{\Omega}\bm{B}^{\top})\).

Properties of the marginal and conditional distributions are reproduced here from \citet{gupta2018matrix}.  These properties are necessary for evaluating the predictive density of \(\tilde{\bm{y}}_{\tilde{p}}\) for \(\tilde{p} < p\) and obtaining conditional draws of \(\bm{Y}^{-}\) given \(\bm{Y}^{+}\) in Algorithm \ref{alg:training}.  For the random variable \(\bm{T}_{N \times p} \sim T_{N \times p}(m, \bm{M}_{N\times p}, \bm{\Sigma}_{N\times N}, \bm{\Omega}_{p \times p})\), arbitrarily partition the matrices as follows:

\begin{equation}
\begin{split}
\bm{T}_{N \times p} = \raisebox{-0.3\baselineskip}{$\begin{blockarray}{cc}
\begin{block}{[c]c}
\bm{T}_{1r} \bigstrut[t]& N_1 \\
\bm{T}_{2r} \bigstrut[b]& N_2 \\
\end{block}
\end{blockarray}$} = \raisebox{-0.666666\baselineskip}{$\begin{blockarray}{cc}
\begin{block}{[cc]}
\bm{T}_{1c} & \bm{T}_{2c}\bigstrut[t]\bigstrut[b]\\
\end{block}
p_1 & p_2
\end{blockarray}$},  &\  \bm{M}_{N \times p} = \raisebox{-0.3\baselineskip}{$\begin{blockarray}{cc}
\begin{block}{[c]c}
\bm{M}_{1r} \bigstrut[t]& N_1 \\
\bm{M}_{2r} \bigstrut[b]& N_2 \\
\end{block}
\end{blockarray}$} = \raisebox{-0.6666666\baselineskip}{$\begin{blockarray}{cc}
\begin{block}{[cc]}
\bm{M}_{1c} & \bm{M}_{2c}\bigstrut[t]\bigstrut[b]\\
\end{block}
p_1 & p_2
\end{blockarray}$}\\
\bm{\Sigma}_{N \times N} = \raisebox{-0.6666666\baselineskip}{$\begin{blockarray}{ccc}
\begin{block}{[cc]c}
\bm{\Sigma}_{11} & \bm{\Sigma}_{12} \bigstrut[t] & N_1\\
\bm{\Sigma}_{21} & \bm{\Sigma}_{22} \bigstrut[b] & N_2\\
\end{block}
N_1 & N_2 & \\
\end{blockarray}$}, & \  \bm{\Omega}_{p \times p} = \raisebox{-0.6666666\baselineskip}{$\begin{blockarray}{ccc}
\begin{block}{[cc]c}
\bm{\Omega}_{11} & \bm{\Omega}_{12} \bigstrut[t] & p_1\\
\bm{\Omega}_{21} & \bm{\Omega}_{22} \bigstrut[b] & p_2\\
\end{block}
p_1 & p_2 & \\
\end{blockarray}$}
\end{split}
\end{equation}
Then the following marginal and conditional distributions of the partitions of \(\bm{T}\) are distributed as
\begin{align}
\bm{T}_{2r} &\sim T_{N_2, p}(m, \bm{M}_{2r}, \bm{\Sigma}_{22}, \bm{\Omega}_{p\times p})  \\
\begin{split}
\bm{T}_{1r}|\bm{T}_{2r} &\sim T_{N_1, p}(m + N_2, \bm{M}_{1r} + \bm{\Sigma}_{12}\bm{\Sigma}^{-1}_{22}(\bm{T}_{2r} - \bm{M}_{2r}), \\
& \phantom{XXXXXXX} \bm{\Sigma}_{11} - \bm{\Sigma}_{12}\bm{\Sigma}_{22}^{-1}\bm{\Sigma}_{21} \bm{\Omega}_{p\times p} + (\bm{T}_{2r} -\bm{M}_{2r})^{\top}\bm{\Sigma}_{22}^{-1}(\bm{T}_{2r}-\bm{M}_{2r}))
\end{split}\\
& \phantom{XX} \notag \\
\bm{T}_{2c} &\sim T_{N, p_2}(m, \bm{M}_{2c}, \bm{\Sigma}_{N\times N}, \bm{\Omega}_{22})  \\
\begin{split}
\bm{T}_{1c}|\bm{T}_{2c} &\sim T_{N, p_1}(m + p_2, \bm{M}_{1c} + (\bm{T}_{2c} - \bm{M}_{2c})\bm{\Omega}_{22}^{-1}\bm{\Omega}_{21}, \\
& \phantom{XXXXXXX} \bm{\Sigma}_{N\times N} + (\bm{T}_{2c} -\bm{M}_{2c})\bm{\Omega}_{22}^{-1}(\bm{T}_{2c}-\bm{M}_{2c})^{\top}, \bm{\Omega}_{11} - \bm{\Omega}_{12}\bm{\Omega}_{22}^{-1}\bm{\Omega}_{21})
\end{split}
\end{align}
The last property of the multivariate t-distribution highly relevant to this application is a correction from \citet{kotz2004multivariate}, useful for the Bayesian typicality index detailed in Appendix \ref{sec:bayes-typicality-index}.  Let the \(\tilde{p}\) dimensional random vector \(\tilde{\bm{y}}_{\tilde{p}} \sim t_{\tilde{\nu}}(\tilde{\bm{\mu}}, \tilde{\bm{\Sigma}})\), then the quadratic form \(Z = (\tilde{\bm{y}}_{\tilde{p}} - \tilde{\bm{\mu}})^{\top}\tilde{\bm{\Sigma}}^{-1}(\tilde{\bm{y}}_{\tilde{p}} - \tilde{\bm{\mu}})/\tilde{p}\), conditioned on \(\tilde{\bm{\mu}}\) is \(F\) distributed as \({\mathbb{F}}(\tilde{p}, \tilde{v})\).

%%%%%%%%%%%%%%%%%%%%%%%%%%%%%%

\section{Decision Theory}

\subsection{Binary Decisions}\label{sec:binary-decision}

When a binary decision needs to be made, some simplifications can be utilized to reduce computation time relative to categorization, especially for large \(K\).  One simplification is the result of
\begin{align}
\sum_{k = 2}^Kp(\tilde{z}_k = 1|\tilde{\bm{y}}_{\tilde{p}}, \bm{Y}_{N \times p}, \bm{\eta}, \bm{\Psi}, \bm{\nu}, \bm{\alpha}) &=  p(\tilde{z}_1 \neq 1|\tilde{\bm{y}}_{\tilde{p}}, \bm{Y}_{N \times p}, \bm{\eta}, \bm{\Psi}, \bm{\nu}, \bm{\alpha}) \notag \\
&= 1- p(\tilde{z}_1 = 1|\tilde{\bm{y}}_{\tilde{p}}, \bm{Y}_{N \times p}, \bm{\eta}, \bm{\Psi}, \bm{\nu}, \bm{\alpha}).
\end{align}
Therefore, under the binary action space of deciding if \(\tilde{\bm{y}}_{\tilde{p}}\) belongs in the first out of \(K\) categories, we only need to know \(p(\tilde{z}_1 = 1|\tilde{\bm{y}}_{\tilde{p}}, \bm{Y}_{N \times p}, \bm{\eta}, \bm{\Psi}, \bm{\nu}, \bm{\alpha})\).  Then, given the loss matrix \(\bm{C}_{2 \times 2}\) 
\begin{equation}
\bm{C}_{2 \times 2} = 
\begin{blockarray}{cccl}
 & a_1  & a_2 & \\
\begin{block}{c[cc]l}
\tilde{z}_1 = 1 & C_{1,1}  & C_{1,2}\bigstrut[t] & \\
\tilde{z}_1 \neq 1 & C_{2,1} & C_{2,2}\bigstrut[b] &, \\
\end{block}
\end{blockarray}
\end{equation}
the posterior expected loss becomes 
\begin{equation}
\begin{split}
\mathbb{E}[L(a_1, \tilde{z}_1)] &=  \mathbb{E}[\tilde{z}_k = 1|\tilde{\bm{y}}_{\tilde{p}}, \bm{Y}_{N \times p}, \bm{\theta}]C_{1,1} + (1- \mathbb{E}[\tilde{z}_k = 1|\tilde{\bm{y}}_{\tilde{p}}, \bm{Y}_{N \times p}, \bm{\theta}])C_{2,1}\\
\mathbb{E}[L(a_2, \tilde{z}_1)] &=   \mathbb{E}[\tilde{z}_k = 1|\tilde{\bm{y}}_{\tilde{p}}, \bm{Y}_{N \times p}, \bm{\theta}]C_{1,2} + (1- \mathbb{E}[\tilde{z}_k = 1|\tilde{\bm{y}}_{\tilde{p}}, \bm{Y}_{N \times p}, \bm{\theta}])C_{2,2},
\end{split}
\end{equation}
where \(\bm{\theta} = [\bm{\eta}, \bm{\Psi}, \bm{\nu}, \bm{\alpha}]\).  The posterior expected loss of \(a_1\) can be simplified as
\begin{equation}
\mathbb{E}[L(a_1, \tilde{z}_1)] =  (C_{1,1} -C_{2,1})\mathbb{E}[\tilde{z}_k = 1|\tilde{\bm{y}}_{\tilde{p}}, \bm{Y}_{N \times p}, \bm{\theta}] + C_{2,1},
\end{equation}
and similarly the posterior expected loss of \(a_2\) can be simplified as
\begin{equation}
\mathbb{E}[L(a_2, \tilde{z}_1)] =   (C_{1,2} - C_{2,2})\mathbb{E}[\tilde{z}_k = 1|\tilde{\bm{y}}_{\tilde{p}}, \bm{Y}_{N \times p}, \bm{\theta}] + C_{2,2}.
\end{equation}
Action \(a_1\) is chosen when \(\mathbb{E}[L(a_1, \tilde{z}_1)] <\mathbb{E}[L(a_2, \tilde{z}_1)] \).  Equivalently
\begin{equation}
(C_{1,1} -C_{2,1})\mathbb{E}[\tilde{z}_k = 1|\tilde{\bm{y}}_{\tilde{p}}, \bm{Y}_{N \times p}, \bm{\theta}] + C_{2,1} <  (C_{1,2} - C_{2,2})\mathbb{E}[\tilde{z}_k = 1|\tilde{\bm{y}}_{\tilde{p}}, \bm{Y}_{N \times p}, \bm{\theta}] + C_{2,2}
\end{equation}
Rearranging terms allows for the simplifications
\begin{equation}
\delta(a_1) = 1 \ \mathrm{if}
\begin{cases}
\mathbb{E}[\tilde{z}_k = 1|\tilde{\bm{y}}_{\tilde{p}}, \bm{Y}_{N \times p}, \bm{\theta}] > \frac{C_{2,1} - C_{2,2}}{(C_{1,2} - C_{2,2} +C_{2,1} - C_{1,1})}  & \mathrm{if} \ (C_{1,2} - C_{2,2} +C_{2,1} - C_{1,1}) > 0\\
\mathbb{E}[\tilde{z}_k = 1|\tilde{\bm{y}}_{\tilde{p}}, \bm{Y}_{N \times p}, \bm{\theta}] < \frac{C_{2,1} - C_{2,2}}{(C_{1,2} - C_{2,2} +C_{2,1} - C_{1,1})}  & \mathrm{if} \ (C_{1,2} - C_{2,2} +C_{2,1} - C_{1,1}) < 0\\
\mathrm{undefined} & \mathrm{if} \ (C_{1,2} - C_{2,2} +C_{2,1} - C_{1,1}) = 0
\end{cases}
\end{equation}
Action \(a_1\) has the minimum posterior expected loss, indicated by \(\delta(a_1) = 1\), when the relevant inequality under \(\delta ( a_1  ) \)  is satisfied.  Notably, under 0-1 loss, when \(C_{1,1} = C_{2,2} = 0\) and \(C_{2,1} = C_{1,2} = 1\), \(a_1\) provides minimum expected loss when \(\mathbb{E}[\tilde{z}_k = 1|\tilde{\bm{y}}_{\tilde{p}}, \bm{Y}_{N \times p}, \bm{\theta}] > \frac{1}{2}\).  Binary decision making under 0-1 loss is therefore simple to evaluate, as the only quantity required is \(p(\tilde{z}_1 = 1|\tilde{\bm{y}}_{\tilde{p}}, \bm{Y}_{N \times p}, \bm{\eta}, \bm{\Psi}, \bm{\nu}, \bm{\alpha})\).

\subsection{Posterior Expected Loss with Missing Training Data}\label{sec:expected-loss-appendix}

The notation \(p(\tilde{z}_k = 1|\tilde{\bm{y}}_{\tilde{p}}, \bm{Y}_{N \times p}, \bm{\eta}, \bm{\Psi}, \bm{\nu}, \bm{\alpha})\) has been used to express the general case of the probability of \(\tilde{z}_k = 1\) given all the data and prior parameters.  However, in application, we expect that there are some missing entries in \(\bm{Y}_{N\times p}\), where with generality \(\bm{Y}^{+}\) indicates the observed entries, and \(\bm{Y}^{-}\) indicates the missing entries.  

In Algorithm \ref{alg:training} and Appendix \ref{sec:missing-data-conditionals} details are provided for using MCMC to generate joint samples from \(p(\bm{Y}^{-}_k|\bm{Y}_k^{+}, \bm{\eta}_k, \bm{\Psi}_k, \nu_k)\) for the \(k^{\mathrm{th}}\) event category.  Decisions are not made by directly using the imputed \(\bm{Y}^{-}\) and instead made using only \(\bm{Y}^{+}\) and the prior hyperparameters.  For decisions where there are missing data entries, we therefore need each
\begin{align}
&p(\tilde{z}_k = 1|\tilde{\bm{y}}_{\tilde{p}}, \bm{Y}^{+}, \bm{\eta}, \bm{\Psi}, \bm{\nu}, \bm{\alpha}) \notag \\
&=\int \dots \int p(\tilde{z}_k = 1,\bm{Y}^{-}_1, \dots, \bm{Y}^{-}_K|\tilde{\bm{y}}_{\tilde{p}}, \bm{Y}^{+},\bm{\eta}, \bm{\Psi}, \bm{\nu}, \bm{\alpha})d\bm{Y}_1^{-} \dots d\bm{Y}^{-}_K \notag \\
&= \int \dots \int \frac{p(\tilde{\bm{y}}_{\tilde{p}}|\tilde{z}_k = 1, \bm{Y}^+_k, \bm{Y}^{-}_k, \bm{\theta})p(\tilde{z}_k = 1|\bm{N}, \bm{\theta})}{p(\tilde{\bm{y}}_{\tilde{p}}|\bm{Y}^{+}, \bm{\theta})}p(\bm{Y}^{-}_1, \dots, \bm{Y}^{-}_K|\bm{Y}^{+}, \bm{\theta})d\bm{Y}_1^{-} \dots d\bm{Y}^{-}_K,
\end{align}
where \(\bm{\theta} = [\bm{\eta}, \bm{\Psi}, \bm{\nu}, \bm{\alpha}] \) and \(\bm{N} = [N_1 , \dots , N_K]\), noting that
\begin{equation}
p(\bm{Y}^{-}_1, \dots, \bm{Y}^{-}_K|\bm{Y}^{+}, \bm{\theta}) = \prod_{k=1}^K p(\bm{Y}^{-}_k|\bm{Y}^{+}_k, \bm{\theta}_k)
\end{equation}
for \(\bm{\theta}_k = [\bm{\eta}_k, \bm{\Psi}_k, \nu_k] \).

\subsection{Bayesian Typicality Index}\label{sec:bayes-typicality-index}

In calculating the typicality index \citep{mclachlan2005discriminant} using B-ECM, one has to account for the fact that \(p(\tilde{\bm{y}}_{\tilde{p}}|\tilde{z}_k = 1, \bm{Y}_{N_k \times p}, \bm{\eta}_k, \bm{\Psi}_k, \nu_k)\) is a multivariate t-distribution, instead of a multivariate normal distribution as in classical ECM.  We know that for the density \(p(\tilde{\bm{y}}_{\tilde{p}}|\tilde{z}_k = 1, \bm{Y}_{N_k \times p}, \bm{\eta}_k, \bm{\Psi}_k, \nu_k)\), the quadratic form \((\tilde{\bm{y}}_{\tilde{p}} - \tilde{\bm{\mu}}_k)^{\top}\tilde{\bm{\Sigma}}_k^{-1}(\tilde{\bm{y}}_{\tilde{p}} - \tilde{\bm{\mu}}_k)/\tilde{p}\) conditioned on \(\tilde{\bm{\mu}}_k\) is \(F\) distributed as \(\mathbb{F}(\tilde{p}, \tilde{\nu}_k)\)\citep{kotz2004multivariate}.  Let \(Z =(\tilde{\bm{y}}_{\tilde{p}} - \tilde{\bm{\mu}}_k)^{\top}\tilde{\bm{\Sigma}}_k^{-1}(\tilde{\bm{y}}_{\tilde{p}} - \tilde{\bm{\mu}}_k)/\tilde{p}\). Then the upper tail probability under the null F-distribution is given by
\begin{equation}\label{eq:pval}
\begin{split}
\int_Z^{\infty} f(x) dx &= F(\infty) - F(Z)\\
& = 1- F(Z),
\end{split}
\end{equation}
where \(f(\cdot)\) is the \(\mathbb{F}(\tilde{p},\tilde{\nu}_k)\) density function and \(F(\cdot)\) is the \(\mathbb{F}(\tilde{p},\tilde{\nu}_k)\) cumulative distribution function.  \(F(Z)\) is the regularized incomplete beta function in the form of 
\begin{equation}\label{eq:reg-incomplete-beta}
I_{\tilde{b}}\left(\frac{\tilde{p}}{2}, \frac{\tilde{\nu}_k}{2}\right) = \frac{\int_{0}^{\tilde{b}} x^{\frac{\tilde{p}}{2} - 1}(1-x)^{\frac{\tilde{\nu}_k}{2} - 1} dx}{\int_{0}^{1} x^{\frac{\tilde{p}}{2} - 1}(1-x)^{\frac{\tilde{\nu}_k}{2} - 1} dx},
\end{equation}
where \(\tilde{b} = (\tilde{\bm{y}}_{\tilde{p}} - \tilde{\bm{\mu}}_k)^{\top}\tilde{\bm{\Sigma}}_k^{-1}(\tilde{\bm{y}}_{\tilde{p}} - \tilde{\bm{\mu}}_k)/((\tilde{\bm{y}}_{\tilde{p}} - \tilde{\bm{\mu}}_k)^{\top}\tilde{\bm{\Sigma}}_k^{-1}(\tilde{\bm{y}}_{\tilde{p}} - \tilde{\bm{\mu}}_k) + \tilde{\nu}_k)\).

With \eqref{eq:reg-incomplete-beta}, we can calculate the p-value \eqref{eq:pval} related to \(\tilde{\bm{y}}_{\tilde{p}}\) given \(\bm{Y}_{N_k \times p}\) and the prior hyperparameters \(\bm{\eta}_k, \bm{\Psi}_k\), and \(\nu_k\).  When all training data is complete, Equation (\ref{eq:reg-incomplete-beta}) evaluates to a point that can be used to judge the typicality of \(\tilde{\bm{y}}_{\tilde{p}}\) in event category \(k\).  However, when \(\bm{Y}_{N_k \times p}\) contains partial observations, there is uncertainty related to \(\bm{Y}^{-}_k\), and therefore a distribution of the p-values.  Samples from the distribution of p-values are readily obtained by plugging the MC samples of \(\bm{Y}^{-}_k\) into Equation (\ref{eq:reg-incomplete-beta}).  Let \(\Phi\) indicate a random p-value from the distribution of p-values \(p(\Phi)\).  

Next, we must decide to either reject or fail to reject \(\tilde{\bm{y}}_{\tilde{p}}\) as being typical of event category \(k\), given the user specified significance level \(\tilde{\alpha}\).  When no data in \(\bm{Y}_{N_k \times p}\) is missing, the next steps are the same as classical typicality and hypothesis testing.  When there is a distribution on \(\Phi\), we have to decide what action to take for hypothesis testing, given \(p(\Phi)\) and a loss function.  Here we examine the difference between choosing between the actions  \(a_1\) and \(a_2\), rejection and failure to reject respectively, when the decision is made using minimum expected loss and simply testing if \(\mathbb{E}[\Phi] < \tilde{\alpha}\).

First, observe that for 0,1 loss matrix
\begin{equation}\label{eq:typicality-mat-1}
\bm{C}^{\Phi}_{2 \times 2} = 
\begin{blockarray}{cccl}
 & a_1  & a_2 & \\
\begin{block}{c[cc]l}
\Phi < \tilde{\alpha} & 0  & 1\bigstrut[t] & \\
\Phi \geq \tilde{\alpha} & 1 & 0\bigstrut[b] &, \\
\end{block}
\end{blockarray}
\end{equation}
the loss function of action \(a_1\) is 
\begin{equation}
\begin{split}
L(\Phi, a_1) &= I[\Phi \geq \tilde{\alpha}]\\
&= I[1-F(Z) \geq \tilde{\alpha}], 
\end{split}
\end{equation}
where \(I[\cdot]\) is the indicator function equal to 1 when the statement inside the brackets is true.  The expected loss, equivalent to the expectation of a function of a random variable, is then
\begin{equation}
\begin{split}
\mathbb{E}\left[L(\Phi, a_1)\right] &=  \int_{0}^{1}I[\Phi \geq \tilde{\alpha}] p(\Phi) d\Phi\\
&= 0\times \int_{0}^{\tilde{\alpha}} p(\Phi)d\Phi + 1 \times \int_{\tilde{\alpha}}^{1} p(\Phi) d\Phi\\
&= Pr(\Phi \geq \tilde{\alpha}).
\end{split}
\end{equation}
Similarly, the expected loss of \(a_2\) is
\begin{equation}
\begin{split}
\mathbb{E}\left[L(\Phi, a_2)\right] &=  \int_{0}^{1}I[\Phi < \tilde{\alpha}] p(\Phi) d\Phi\\
&= 1\times \int_{0}^{\tilde{\alpha}} p(\Phi)d\Phi + 0 \times \int_{\tilde{\alpha}}^{1} p(\Phi) d\Phi\\
&= Pr(\Phi < \tilde{\alpha}).
\end{split}
\end{equation}
This implies that action \(a_1\), rejection, is the minimum loss action when \(Pr(\Phi \geq \tilde{\alpha}) < Pr(\Phi < \tilde{\alpha})\).  Which is equivalent to the condition \(\frac{1}{2} < Pr(\Phi < \tilde{\alpha})\).  \(Pr(\Phi < \tilde{\alpha})\) is obtained from the CDF function of \(p(\Phi)\) evaluated at \(\tilde{\alpha}\), subsequently noted as \(F^{\Phi}(\tilde{\alpha})\).

One may be tempted to evaluate \(\mathbb{E}[\Phi] < \tilde{\alpha}\) as the decision criterion for rejection, using the loss matrix
\begin{equation}\label{eq:typicality-mat-2}
\bm{C}^{\Phi}_{2 \times 2} = 
\begin{blockarray}{cccl}
 & a_1  & a_2 & \\
\begin{block}{c[cc]l}
\mathbb{E}[\Phi] < \tilde{\alpha} & 0  & 1\bigstrut[t] & \\
\mathbb{E}[\Phi] \geq \tilde{\alpha} & 1 & 0\bigstrut[b] & . \\
\end{block}
\end{blockarray}
\end{equation}
However, this second set of decision criteria will not always lead to the same action as in \eqref{eq:typicality-mat-1}.

\begin{lemma}
The loss functions stipulated by \eqref{eq:typicality-mat-1} and \eqref{eq:typicality-mat-2} do not produce equivalent minimum expected loss for the action space, where \(\Phi\) is distributed as \(p(\Phi)\).
\end{lemma}

\begin{proof}
For the minimum expected loss actions to be consistent between \eqref{eq:typicality-mat-1} and \eqref{eq:typicality-mat-2}, there would be no density of \(p(\Phi)\), and no values of \(\tilde{\alpha}\) where the simultaneous conditions \(F^{\Phi}(\tilde{\alpha}) =  \int_0^{\tilde{\alpha}} p(\Phi) d\Phi > 0.5\) and  \(\mathbb{E}[\Phi] \ge \tilde{\alpha}\) occur.  Such conditions would mean \(a_1\) produces minimum loss for \eqref{eq:typicality-mat-1} and \(a_2\) produces minimum loss for \eqref{eq:typicality-mat-2}.

Under the following conditions we can observe this inconsistency of the loss functions.  Let \(p(\Phi) = \mathrm{Beta}(2,5)\), and \(\tilde{\alpha} = 0.2687\).  In this case, \(F^{\Phi}(\tilde{\alpha}) \approx 0.51 > 0.5\) and \(\mathbb{E}[\Phi] = 2/7 \approx 0.2857 > \tilde{\alpha}\).
\end{proof}
As a more general note; the definition of expectation for real valued random variables is \(\int_0^{\mathbb{E}[\Phi]} F(\Phi)d\Phi = \int_{\mathbb{E}[\Phi]}^1 (1-F(\Phi))d\Phi\), and not necessarily where \(F(\Phi) = 0.5\).  That is, the mean does not necessarily equal the median. 

We choose the first decision criterion, because this criterion is invariant to a one-to-one bijective transformation \(g(\Phi)\) of the random variable \(\Phi\).  As long as the \(F^{g(\Phi)}(\tilde{\alpha})\) can be evaluated, the loss function will be consistent. However, \(\mathbb{E}[g(\Phi)] < g(\tilde{\alpha})\) may be inconsistent with the condition \(\mathbb{E}[\Phi] < \tilde{\alpha}\).  

In this case, evaluating \(F^{\Phi}(\tilde{\alpha})\) is implemented using the empirical CDF of the \(T-B\) Monte Carlo samples
\begin{equation}
F^{\Phi}(\tilde{\alpha}) \approx \frac{\sum_{t = 1}^{T-B}I[\Phi^{(t)} < \tilde{\alpha}]}{T-B},
\end{equation}
which should work well for 0,1 loss where we are evaluating if \(F^{\Phi}(\tilde{\alpha}) < 0.5\), but may require large \(T-B\), or other techniques, for loss functions which require evaluation of the CDF \(F^{\Phi}(\tilde{\alpha}) \gtrapprox 0\) or  \(F^{\Phi}(\tilde{\alpha}) \lessapprox 1\).

\section{Monte Carlo Experiments}

\subsection{Synthetic Data}\label{sec:syn-data-appendix}

\begin{algorithm}
\caption{Data generation for a single Monte Carlo experiment using synthetic data.}\label{alg:syn-data}
\begin{algorithmic}[1]
\Require \(p, N^{\mathrm{train}}, N^{\mathrm{train}^{-}}, N^{\mathrm{test}}, o\)
\State \([\bm{\mu}^{\star}_1, \bm{\mu}^{\star}_2, \bm{\mu}^{\star}_3]\sim \mathcal{N}_{p,3}(\bm{0}_{p \times 3}, 0.25 \times \mathbb{I}_{p} \otimes \mathbb{I}_3 ) \)
\State \([N_1, N_2, N_3] \sim \mathrm{Multinomial}(N_{\mathrm{train}} + N_{\mathrm{train}^{-}} + N_{\mathrm{test}}, [\frac{1}{3}, \frac{1}{3}, \frac{1}{3}])\)
\State \(k^{\star} \sim \mathrm{Categorical}([\frac{1}{3}, \frac{1}{3}, \frac{1}{3}])\)
\For{\(k \in \{1, 2, 3\}\)}
    \State \(x \sim \mathrm{Bernoulli}(0.5)\)
    \State \(b \gets 2-x \)
    \State \(\bm{\Sigma}_k^{\star} \sim \mathcal{W}^{-1}(p + 4, \mathbb{I}_p)\)
    \If{\(b = 2\)}
        \State \(\mathrm{dim}(\bm{b}^1) \sim \mathrm{Categorical}(\frac{1}{p-1} \bm{1}_{p-1}^{\top})\)
        \State \(\mathrm{dim}(\bm{b}^2) \gets p - \mathrm{dim}(\bm{b}^1)\)
        \State \(\bm{b}^1 \sim \mathrm{sample} \ \mathrm{dim}(\bm{b}^1) \ \mathrm{elements \ from \ } \{1, \dots , p\}\) without replacement
        \State \(\bm{b}^2 \gets \{1 , \dots , p\} \notin \bm{b}^1\)
        \State \((\bm{\Sigma}_k^{\star})_{[\bm{b}^1, \bm{b}^2]} \gets 0\)
        \State \((\bm{\Sigma}_k^{\star})_{[\bm{b}^2, \bm{b}^1]} \gets 0\)
    \EndIf
    \State \(\bm{Y}_{N_k \times p} \sim \mathcal{N}_{N_k, p}(\bm{{1}}_{N_k}(\bm{\mu}^{\star}_k)^{\top}, \mathbb{I}_{N_k} \otimes \bm{\Sigma}_k^{\star})\)
    \State \(\bm{Y}_{N_k \times p} \gets \mathrm{logit}^{-1}(\bm{Y}_{N_k \times p})\)
\EndFor
\State \(\mathcal{S}^{\mathrm{all}} \equiv  \{1^{(1)} , \dots , N_1^{(1)}, 1^{(2)}, \dots , N_2^{(2)}, \dots , 1^{(3)}, \dots , N_3^{(3)} \}\)
\While{\(\mathrm{Any} \ \mathrm{dim}(\mathcal{S}^{\mathrm{train}}_{k} : k \in \{1,2,3\}) < 2\)}
    \State \(\mathcal{S}^{\mathrm{test}} \sim \mathrm{sample} \ N^{\mathrm{test}} \ \mathrm{elements \ from} \ \mathcal{S}^{\mathrm{all}} \) without replacement
    \State \(\mathcal{S}^{\mathrm{train}} \gets \mathcal{S}^{\mathrm{all}} \notin \mathcal{S}^{\mathrm{test}}\)
    \State \(\mathcal{S}^{\mathrm{train}^{-}} \sim \mathrm{sample} \ N^{\mathrm{train}^{-}} \ \mathrm{elements \ from \ } \mathcal{S}^{\mathrm{train}}\) without replacement
    \State \(\mathcal{S}^{\mathrm{train}} \gets \mathcal{S}^{\mathrm{all}} \notin \{\mathcal{S}^{\mathrm{test}}, \mathcal{S}^{\mathrm{train}^{-}}\}\)
\EndWhile
\State \(\bm{Y}_{N^{\mathrm{train}^{-}}_k \times p} \gets (\bm{Y}_{N_k \times p})_{[\mathcal{S}^{\mathrm{train}^{-}}_k, \cdot]}: \forall k \in \{1, 2, 3\}\)
\State \(\bm{Y}_{N^{\mathrm{train}^{-}} \times p} \gets \{\bm{Y}_{N_1^{\mathrm{train}^{-}} \times p}, \bm{Y}_{N_2^{\mathrm{train}^{-}}\times p}, \bm{Y}_{N_3^{\mathrm{train}^{-}}\times p}\}\)
\For{\(i \in \{ 1, \ldots, N^{\mathrm{train}^{-}}\}\)}
\State \(\mathcal{P}^{\star +}_i \sim \mathrm{Categorical}(\frac{1}{p}\bm{1}_p^{\top})\)
\State \(\mathcal{P}^{-}_i = \{ 1, \ldots, p \} - \mathcal{P}^{\star +}_i\)
\State \(\mathcal{P}^{\star -}_i \sim \mathrm{sample \ one \ element \ from \ } \mathcal{P}^{-}_i\)
\State \(\mathcal{P}^{-}_i = \{ 1, \ldots, p \} - \{\mathcal{P}^{\star +}_i, \mathcal{P}^{\star -}_i\}\)
\State \(\underline{\mathcal{P}}^{-}_i = \mathrm{sample \ } \lfloor{(p-2) \times o\rfloor} \mathrm{\ elements \ from \ } \mathcal{P}^{-}_i\) without replacement
\State \((\bm{Y}_{N^{\mathrm{train}^{-}} \times p})_{[i,\underline{\mathcal{P}}^{-}_i]} \gets \)\texttt{NA}
\EndFor
\algstore{syndata}
\end{algorithmic}
\end{algorithm}

\begin{algorithm}                     
\begin{algorithmic} [1]              
\algrestore{syndata}
\State \(\bm{Y}_{N^{\mathrm{test}}_k \times p} \gets (\bm{Y}_{N_k \times p})_{[\mathcal{S}^{\mathrm{test}}_k]}: \forall k \in \{1, 2, 3\}\)
\State \(\bm{Y}_{N^{\mathrm{test}} \times p} \gets \{\bm{Y}_{N_1^{\mathrm{test}} \times p}, \bm{Y}_{N_2^{\mathrm{test}}\times p}, \bm{Y}_{N_3^{\mathrm{test}}\times p}\}\)
\For{\(i \in \{ 1, \ldots, N^{\mathrm{test}}\}\)}
\State \(\mathcal{P}^{\star +}_i \sim \mathrm{Categorical}(\frac{1}{p}\bm{1}_p^{\top})\)
\State \(\mathcal{P}^{-}_i = \{ 1, \ldots, p \} - \mathcal{P}^{\star +}_i\)
\State \(\underline{\mathcal{P}}^{-}_i = \mathrm{sample \ } \lfloor{(p-1) \times o\rfloor} \mathrm{\ elements \ from \ } \mathcal{P}^{-}_i\) without replacement
\State \((\bm{Y}_{N^{\mathrm{test}} \times p})_{[i,\underline{\mathcal{P}}^{-}_i]} \gets \)\texttt{NA}
\EndFor
\State \(\bm{Y}_{N^{\mathrm{train}}_k \times p} \gets (\bm{Y}_{N_k \times p})_{[\mathcal{S}^{\mathrm{train}}_k, \cdot]}: \forall k \in \{1, 2, 3\}\)
\State \(\bm{Y}_{N^{\mathrm{train}} \times p} \gets \{\bm{Y}_{N_1^{\mathrm{train}} \times p}, \bm{Y}_{N_2^{\mathrm{train}}\times p}, \bm{Y}_{N_3^{\mathrm{train}}\times p}\}\)
\end{algorithmic}
\end{algorithm}

Additional notation is required to describe the algorithm used for generating synthetic data.  Much of the notation is the same as the other text provided, however in some cases it was necessary to change the notation in order to provide a better description for this context.  The fraction of data missing from \(\bm{Y}_{N^{\mathrm{train}^{-}} \times p}\) and \(\bm{Y}_{N^{\mathrm{test}} \times p}\) is \(o\).  When missing observations are generated, the code ensures one of the discriminants is guaranteed to remain in order to reduce computation time relative to random deletion where the possibility of deleting all data for a single observation is a possibility.  The full set of training data \(\{\bm{Y}_{N^{\mathrm{train}} \times p }, \bm{Y}_{N^{\mathrm{train}^{-}} \times p }\}\) has a proportion of missing data less than \(o\).

The additional nomenclature follows.  \(N^{\mathrm{train}}, N^{\mathrm{train}^{-}}\), and \( N^{\mathrm{test}}\) are the integer sizes of the total training set of full observations, training set of partial observations, and testing set respectively, which remain constant across MC iterations.  Adding the \(k\) subscript, as in \(N^{\mathrm{train}}_k\) indicates the integer number specific to the \(k^{\mathrm{th}}\) event category. \(\bm{\mu}^{\star}_k\) is the true mean for the \(k^{\mathrm{th}}\) event category.  \(\mathcal{N}_{p,n}(\bm{M}, \bm{\Sigma} \otimes \bm{\Psi})\) is the matrix variate normal distribution \citep{gupta2018matrix}.  \(N_k\), within the algorithm, is the total number of training and testing observations in the \(k^{\mathrm{th}}\) event category.  \(\mathrm{Categorical}(\bm{p}^{\top})\) defines \(N\) draws from the categorical distribution, the number of categories is defined by the length of the vector of probabilities \(\bm{p}\).  \(k^{\star}\) is the randomly selected category of interest, used to define the binary decision problem, calculating false negatives, and calculating false positives.

\(\mathcal{W}^{-1}(m, \bm{\Psi})\) is the invese Wishart distribution, also known as the inverted Wishart distribution \citep{gupta2018matrix}.  \(\mathrm{dim}(\bm{x})\) defines the dimension of a vector \(\bm{x}\).  The random variable \(b\) indicates the number of independent blocks in random covariance matrix \(\bm{\Sigma}^{\star}_k\).  If \(b = 2\), then the random vectors \(\bm{b}^1\) and \(\bm{b}^2\) index the elements of \(\bm{\Sigma}^{\star}_k\) which correspond to each block.  

The notation \((\bm{X})_{[a,b]}\) indicates a rectangular subset of the matrix \(\bm{X}\), where the selected rows are indicated by \(a\) and the selected columns are indicated by \(b\).  When the subset includes all columns, then the subset is noted as \((\bm{X})_{[a,\cdot]}\), and when a subset of a vector is taken, only one dimension is given in the brackets. \(\mathrm{logit}^{-1}(x)\) denotes the inverse logit function of \(x\), where \(\mathrm{logit}^{-1}(x) = 1/(1 + \exp \{- x\})\).  

The set \(\mathcal{S}^{\mathrm{all}}\) indexes all data in all event categories.  Similarly, \(\mathcal{S}^{\mathrm{train}}\), \(\mathcal{S}^{\mathrm{train}^{-1}}\), and  \(\mathcal{S}^{\mathrm{test}}\) index the training and testing sets.  Adding the \(k\) subscript, as in \(\mathcal{S}_k^{\mathrm{train}}\) indicates indexing specific to the \(k^{\mathrm{th}}\) event category.  The notation \(i^{(k)}\) used in defining \(\mathcal{S}^{\mathrm{all}}\) indicates the \(i^{\mathrm{th}}\) of \(N_k\) observations generated for the \(k^{\mathrm{th}}\) event category.

Algorithm steps using any variant of the notation \(\mathcal{P}\) are related to a particular discriminant.  The subscript \(i\) used in any variant of \(\mathcal{P}_i\) indexes observation \(i\).  \(\mathcal{P}_i^{\star+}\) denotes a discriminant in observation \(i\) which is guaranteed to be part of the final data set. \(\mathcal{P}_i^{-}\) is the set of discriminants in observation \(i\) which have a chance to be missing.  \(\mathcal{P}_i^{\star -}\) is a random sample from \(\mathcal{P}_i^{-}\) which is designated as guaranteed to be missing.  \(\underline{\mathcal{P}}^{-}_i\) is the set of discriminants randomly designated to be missing for observation \(i\) within the final data set.

\subsubsection{Reproducable Code}\label{sec:synthetic-data-code}

In an effort to improve reproducibility, the code used to produce Fig. \ref{fig:syn-boxplot} is provided as a vignette within the \textsf{R} package \texttt{ezECM}.  In order to access the code first install the development version of the package from GitHub with the command \\ \texttt{remotes::install\_github("lanl/ezECM")}, or from CRAN using \\ \texttt{install.packages("ezECM")}.  Then, load the package with the command \texttt{library(ezECM)}.  Last, display the vignette with \texttt{vignette("syn-data-code", package = "ezECM")}, where the code used for generating the data and statistical inference can be inspected.

\subsection{Seismic Discriminant Experiment}\label{sec:dale-data-appendix}

\begin{algorithm}
\caption{Data set generation for a single Monte Carlo simulation using \textbf{Seismic Discriminant Data}}\label{alg:dale-data}
\begin{algorithmic}[1]
\Require \(\bm{Y}_{N^{-} \times p}, \bm{Y}_{N^{+} \times p}\)
\State \(\mathcal{S}^{+} \equiv  \{1^{(1)} , \dots , (N_{\texttt{EX}}^{+})^{(1)}, 1^{(2)}, \dots , (N_{\texttt{SEQ}}^{+})^{(2)}, \dots , 1^{(3)}, \dots , (N_{\texttt{DEQ}}^{+})^{(3)} \}\)
\State \(\mathcal{S}^{\mathrm{train}} \sim \mathrm{sample \ }\lfloor 0.8 \times ( N^{+}_{\texttt{EX}} + N^{+}_{\texttt{SEQ}}) \rfloor \mathrm{\ elements \ from \ } \mathcal{S}^{+} - \{ 1^{(3)}, \dots , (N_{\texttt{DEQ}}^{+})^{(3)} \}\) without replacement
\State \(\mathcal{S}^{\mathrm{train}} \gets \{\mathcal{S}^{\mathrm{train}},\{ 1^{(3)}, \dots , (N_{\texttt{DEQ}}^{+})^{(3)} \} \}\)
\State \(\mathcal{S}^{-} \equiv  \{1^{(1)} , \dots , (N_{\texttt{EX}}^{-})^{(1)}, 1^{(2)}, \dots , (N_{\texttt{SEQ}}^{-})^{(2)}, \dots , 1^{(3)}, \dots , (N_{\texttt{DEQ}}^{-})^{(3)} \}\)
\State \(\mathcal{S}^{\mathrm{train}^{-}} \sim \mathrm{sample \ } \left\lfloor (N_{\texttt{EX}}^{-} + N_{\texttt{SEQ}}^{-} + N_{\texttt{DEQ}}^{-}) \times \sqrt{p} /(1 + \sqrt{p}) \right\rfloor \mathrm{\ elements \ from \ } \mathcal{S}^{-} \) without replacement
\State \(\mathcal{S}^{\mathrm{test}} \gets \{\mathcal{S}^{-} - \mathcal{S}^{\mathrm{train}^{-}}, \mathcal{S}^{+} - \mathcal{S}^{\mathrm{train}}\}\)
\State \(\bm{Y}_{N^{\mathrm{train}} \times p} \gets (\bm{Y}_{N^{+} \times p})_{[\mathcal{S}^{\mathrm{train}}, \cdot]}\)
\State \(\bm{Y}_{N^{\mathrm{train}^{-}} \times p} \gets (\bm{Y}_{N^{-} \times p})_{[\mathcal{S}^{\mathrm{train}^{-}}, \cdot]}\)
\State \(\bm{Y}_{N^{\mathrm{test}} \times p} \gets \{\bm{Y}_{N^{-} \times p}, \bm{Y}_{N^{+} \times p}\}_{[\mathcal{S}^{\mathrm{test}}, \cdot]}\)
\end{algorithmic}
\end{algorithm}

In this real data set, the elements with missing data do not have to be selected and are truly missing.  Three usable event categories are present in the data set; explosions, shallow earthquakes, and deep earthquakes, denoted as \texttt{EX}, \texttt{SEQ}, and \texttt{DEQ} respectively.  \(\mathcal{S}^{+}\) indexes the fully observed data \(\bm{Y}_{N^{+} \times p}\), and \(\mathcal{S}^{-}\) indexes the partially observed data \(\bm{Y}_{N^{-} \times p}\).  The \texttt{DEQ} category had only enough full observations to train the C-ECM model and therefore all full observations are included in training in line 3.  The remainder of the notation is the same as what is used in Appendix \ref{sec:syn-data-appendix}.

\subsubsection{Data}\label{sec:data}

After installation of the \texttt{ezECM} package, a global variable with the data used in the experiment can be specified using the function \texttt{read.csv(system.file("extdata", "ECM\_validation\_data\_scrubbed.csv", package="ezECM"))[,c("Source.Type", "pLP", "pCF\_EKA", "pCF\_GBA", "pCF\_WRA", "pCF\_YKA")]}. Then the final data set can be obtained by removing any event rows where all discriminants are missing as well as any \texttt{Source.Type} of \texttt{"MEX"}.

\end{document}